\documentclass{article}

\usepackage{microtype}
\usepackage{graphicx}
\usepackage{subfigure}
\usepackage{booktabs}
\usepackage{hyperref}

\usepackage[utf8]{inputenc}
\usepackage[T1]{fontenc}
\usepackage{amsmath,amsthm,amssymb}
\usepackage{cleveref}
\usepackage{mathtools}
\usepackage{xcolor}
\allowdisplaybreaks

\newcommand{\argmax}{\mathrm{argmax}}
\newcommand{\argmin}{\mathrm{argmin}}

\newcommand{\supp}{\mathrm{supp}}
\newcommand{\bfzero}{\mathbf{0}}

\newcommand{\rme}{\mathrm{e}}
\newcommand{\rmO}{\mathrm{O}}
\newcommand{\bbR}{\mathbb{R}}
\newcommand{\bbZ}{\mathbb{Z}}
\newcommand{\bfA}{\mathbf{A}}
\newcommand{\bfB}{\mathbf{B}}
\newcommand{\bfe}{\mathbf{e}}
\newcommand{\bfw}{\mathbf{w}}
\newcommand{\bfx}{\mathbf{x}}
\newcommand{\bfy}{\mathbf{y}}
\newcommand{\bfz}{\mathbf{z}}
\newcommand{\calF}{\mathcal{F}}
\newcommand{\calI}{\mathcal{I}}
\newcommand{\calM}{\mathcal{M}}
\newcommand{\calN}{\mathcal{N}}
\newcommand{\calP}{\mathcal{P}}

\newtheorem{theorem}{Theorem}[]
\newtheorem{lemma}[theorem]{Lemma}
\newtheorem{corollary}[theorem]{Corollary}
\newtheorem{proposition}[theorem]{Proposition}
\theoremstyle{definition}
\newtheorem{definition}[theorem]{Definition}
\newtheorem{remark}[theorem]{Remark}

\usepackage[accepted]{icml2020}

\icmltitlerunning{Approximation Guarantees of Local Search Algorithms via Localizability of Set Functions}

\begin{document}

\twocolumn[
\icmltitle{Approximation Guarantees of Local Search Algorithms\\via Localizability of Set Functions}

\begin{icmlauthorlist}
\icmlauthor{Kaito Fujii}{nii}
\end{icmlauthorlist}

\icmlaffiliation{nii}{National Institute of Informatics, Tokyo, Japan}

\icmlcorrespondingauthor{Kaito Fujii}{fujiik@nii.ac.jp}

\icmlkeywords{Machine Learning, ICML}

\vskip 0.3in
]

\printAffiliationsAndNotice{}  

\begin{abstract}
This paper proposes a new framework for providing approximation guarantees of local search algorithms. Local search is a basic algorithm design technique and is widely used for various combinatorial optimization problems. To analyze local search algorithms for set function maximization, we propose a new notion called \textit{localizability} of set functions, which measures how effective local improvement is. Moreover, we provide approximation guarantees of standard local search algorithms under various combinatorial constraints in terms of localizability. The main application of our framework is sparse optimization, for which we show that restricted strong concavity and restricted smoothness of the objective function imply localizability, and further develop accelerated versions of local search algorithms. We conduct experiments in sparse regression and structure learning of graphical models to confirm the practical efficiency of the proposed local search algorithms.
\end{abstract}

\section{Introduction}\label{sec:local-background}
Local search is a widely used technique to design efficient algorithms for optimization problems.
Roughly speaking, local search algorithms start with an initial solution and gradually increase the objective value by repeatedly moving the solution to a nearby point.
While this approach leads to effective heuristics for many optimization problems in practice, it is not always easy to provide approximation guarantees on their performance.

In this paper, we propose a generic framework for providing approximation guarantees of local search algorithms for \textit{set function optimization}.
Set function optimization is a problem of finding an (approximately) optimal set from all feasible sets.
Various machine learning tasks have been formulated as set function optimization problems, such as feature selection \citep{Das2011,Elenberg18}, summarization \citep{LB11,BMKK14}, or active learning \citep{HJZL06,GK11}.

A promising approach to analyze local search algorithms for set function optimization is to utilize \textit{submodularity}.
Submodularity \citep{Fujishige2005} is a property of set functions useful for designing efficient algorithms and has been extensively studied.
Existing studies showed that for maximizing a submodular function subject to a certain constraint, local search procedures yield a constant-factor approximate solution to an optimal solution \citep{NWF78,FNW78,LSV10,FNSW11}.
Local search algorithms have been applied to several machine learning tasks that have submodularity \citep{IJB13,Balkanski2016}, but there are many other practical set functions that deviate from submodularity.

To analyze local search algorithms for these problems, we propose a novel analysis framework that can be applied to local search algorithms for non-submodular functions.
The key notion of our framework is a property of set functions, which we call \textit{localizability}.
Intuitively, localizability is a property that implies any local optimum is a good approximation to a global optimal solution.
By utilizing this property, we show that for maximizing a set function with localizability under a certain constraint, simple local search algorithms achieve a good approximation.

\begin{table*}[t]
	\vskip 0.15in
	\centering
	\caption{
		Comparison of existing bounds on approximation ratios of local search algorithms, greedy algorithms, and modular approximation for sparse optimization with combinatorial constraints.
		The result of \citet{Elenberg2017} is indicated by $\dagger$.
		The result of \citet{CFK18} is indicated by $\ddagger$.
		$M_s$ and $M_{s,t}$ are restricted smoothness constants and $m_s$ is a restricted strong concavity constant (See \Cref{def:restricted} for details).
		$T$ is the number of iterations of local search algorithms and $s$ is the maximum cardinality of feasible solutions.
	}
	\label{table:sparse-ratio}
	\begin{tabular}{lccc}
		\toprule
		Constraint & Local search & Greedy-based & Modular approx.\\
		\midrule
		Cardinality & $\frac{m_{2s}^2}{M_{s,2}^2} \left( 1 - \exp\left(\frac{M_{s,2}T}{sm_{2s}}\right) \right)$ & $1 - \exp\left(- \frac{m_{2s}}{M_{s,1}}\right)$ $\dagger$ & $\frac{m_1 m_{s}}{M_1 M_s}$\\
		\hline
		Matroid & $\frac{m_{2s}^2}{M_{s,2}^2}  \left( 1 - \exp\left(\frac{M_{s,2}T}{sm_{2s}}\right) \right)$ & $\frac{1}{(1 + \frac{M_{s,1}}{m_s})^2}$ $\ddagger$ & $\frac{m_1 m_{s}}{M_1 M_s}$\\
		\hline
		\begin{tabular}{@{}c@{}}$p$-Matroid intersection\\or $p$-Exchange systems\end{tabular} & $\frac{1}{p-1+1/q}\frac{m_{2s}^2}{M_{s,2}^2}  \left( 1 - \exp\left(\frac{(p-1+1/q)M_{s,2}T}{sm_{2s}}\right) \right) $ & N/A & $\frac{1}{p-1+1/q}\frac{m_1 m_{s}}{M_1 M_s}-\epsilon$\\
		\bottomrule
	\end{tabular}
\end{table*}

The main application of our framework is \textit{sparse optimization}.
Sparse optimization is the problem of finding a sparse vector that optimizes a continuous objective function.
It has various applications such as feature selection for sparse regression and structure learning of graphical models.
An approach to sparse optimization is a reduction to a set function optimization problem, which is adopted by \citet{JRD16} and \citet{Elenberg2017}.
We show that localizability of this set function is derived from restricted strong concavity and restricted smoothness of the original objective function, which implies approximation guarantees of local search algorithms.
Furthermore, we devise accelerated variants of our proposed local search algorithms by utilizing the structure of sparse optimization.
An advantage of our approach over existing methods is its applicability to a broader class of combinatorial constraints.

\paragraph{Our contribution.}
In this paper, we propose a new property of set functions called \textit{localizability} and provide a lower bound on the approximation ratio of local search algorithms for maximizing a set function with this property.
Our contribution is summarized as follows.
\begin{itemize}
\item We define localizability of set functions and show that localizability of sparse optimization is derived from restricted strong concavity and restricted smoothness of the original objective function.
\item Under the assumption of localizability, we provide lower bounds on the approximation ratio of a standard local search algorithm under a matroid constraint, $p$-matroid intersection constraint, or $p$-exchange system constraint.
\item For sparse optimization, we propose two accelerated variants of local search algorithms, which we call \textit{semi-oblivious} and \textit{non-oblivious} local search algorithms.
\item We conduct experiments on sparse regression and structure learning of graphical models to confirm the practical efficiency of our accelerated local search algorithms.
\end{itemize}

\subsection{Related Work}

\paragraph{Local search for submodular maximization.}
For monotone submodular maximization, several algorithms have been designed based on local search.
\citet{NWF78} proposed a $1/2$-approximation local search procedure for a cardinality constraint, which they call an \textit{interchange heuristic}.
\citet{FNW78} generalized this result to a single matroid constraint.
For a $p$-matroid intersection constraint, \citet{LSV10} proposed a $(1/p-\epsilon)$-approximation local search algorithm.
\citet{FNSW11} proposed a novel class of constraints called $p$-exchange systems and devised a $(1/p-\epsilon)$-approximation local search algorithm.
\citet{FW14} devised a $(1 - 1/\rme)$-approximation local search algorithm for a matroid constraint.
Also for non-monotone submodular maximization, constant-factor approximation local search algorithms have been devised \citep{FMV11,LMNS09}.
While local search algorithms for submodular maximization have been studied extensively, there are only a few results on those for non-submodular maximization.

\paragraph{Approximation algorithms for non-submodular function maximization and sparse optimization.}
For non-submodular function maximization, many existing studies have adopted an approach based on \textit{greedy algorithms}.
\citet{Das2011} analyzed greedy algorithms for sparse linear regression by introducing the notion of submodularity ratio.
\citet{Elenberg18} extended their results to sparse optimization problems with restricted strong concavity and restricted smoothness.
\citet{CFK18} showed that the random residual greedy algorithm achieves $(\gamma_{} / (1 + \gamma_{}))^2$-approximation for a set function maximization with submodularity ratio $\gamma$ under a single matroid constraint\footnote{They did not specify the subscripts of $\gamma$, but it is not larger than $\min_{i=1,\cdots,s} \gamma_{i-1,s-i}$, where $s$ is the rank of the matroid.}.
We compare our results with existing methods for sparse optimization in \Cref{table:sparse-ratio}.
Note that the results of \citet{Das2011} and \citet{CFK18} hold for any monotone set function whose submodularity ratio is bounded, while we utilize a stronger property derived from restricted strong concavity and restricted smoothness.
\citet{Bian17} analyzed the greedy algorithm for maximizing a set function whose submodularity ratio and generalized curvature are both bounded.
\citet{Sakaue19} considered sparse optimization with a constraint expressed by a monotone set function with bounded superadditivity ratio and restricted inverse curvature, but their framework cannot deal with matroid constraints and $p$-exchange system constraints.
\citet{FS18} developed a similar analysis to ours, but their focus lies in a different problem called dictionary selection.
The Frank-Wolfe algorithm~\citep{FW56,Jaggi13} is a continuous optimization method that is often applied to sparse optimization.
A variant called the \textit{pairwise Frank-Wolfe algorithm}~\citep{LJ15} incorporates the technique of moving weight between two atoms at each iteration, which is similar to our local search procedures, but their guarantees are incomparable to ours.

\paragraph{Modular approximation for sparse optimization.}
For sparse optimization with structured constraints, there exists a trivial benchmark called \textit{modular approximation} \citep{CK11}.
Modular approximation maximizes a linear function that approximates the original set function by ignoring all interactions between elements.
We provide a detailed description and analysis of modular approximation in \Cref{sec:modular-approximation}.

\paragraph{Sparse recovery.}
Many existing studies on sparse optimization focus on sparse recovery guarantees, which cannot be compared directly with our guarantees on approximation ratios.
In the context of compressed sensing, several algorithms similar to our non-oblivious local search algorithms have been developed \citep{NT10,BRB13}.
\citet{KC12} and \citet{BLSGB16} developed frameworks that can be applied to a matroid constraint.
For structure learning of graphical models, \citet{JJR11} provided sparse recovery guarantees for the forward-backward greedy algorithm by assuming restricted strong concavity and restricted smoothness.
Recently, algorithms with recovery guarantees under weaker assumptions have been developed \citep{Bresler15,KM17,WSD19}.

\subsection{Organization}
The rest of this paper is organized as follows.
\Cref{sec:local-setting} specify the problem settings that we tackle in this paper.
\Cref{sec:local-approximate} introduces the notion of localizability and shows localizability of sparse optimization.
In \Cref{sec:local-algorithms}, we propose local search algorithms for a matroid constraint, $p$-matroid intersection constraint, or $p$-exchange system constraint.
In \Cref{sec:local-acceleration}, we devise accelerated local search algorithms for sparse optimization.
In \Cref{sec:local-applications}, we describe applications of our problem settings: sparse regression and structure learning of graphical models.
In \Cref{sec:local-experiments}, we empirically compare our proposed algorithms with existing methods.
Due to space constraints, we defer all proofs to the appendix.

\section{Problem Setting}\label{sec:local-setting}
In this section, we introduce the problem settings that we deal with in this paper.

\paragraph{Set function maximization.}
Let $N \coloneqq [n]$ be the ground set and define a non-negative set function $f \colon 2^N \to \bbR_{\ge 0}$.
Throughout the paper, we assume $f$ is monotone, i.e., $f(X) \le f(Y)$ holds for any $X \subseteq Y \subseteq N$.
We say $f$ is submodular when $f(S \cup \{v\}) - f(S) \ge f(T \cup \{v\}) - f(T)$ for any $S \subseteq T \subseteq N$ and $v \in N \setminus T$.
Let $\calI \subseteq 2^N$ be a set family that represents all feasible solutions.
We assume $(N, \calI)$ is an independence system, that is, $\emptyset \in \calI$ and $X \in \calI$ for any $X \subseteq Y$ such that $Y \in \calI$.
A set function maximization problem can be written as
\begin{equation}\label{eq:set-max}
		\text{Maximize} \quad f(X) \qquad \text{subject to} \quad X \in \calI.
\end{equation}
In general, suppose we have access to a value oracle and independence oracle, which return the value of $f(X)$ and the Boolean value that represents whether $X \in \calI$ or not for any input $X \in N$, respectively.

We consider three classes of independence systems: matroid constraints, $p$-matroid intersection, and $p$-exchange systems, which include structures that appear in applications.
A standard setting of sparse optimization where $\calI = \{ X \subseteq N \mid |X| \le s \}$ is a special case of matroid constraints.
\begin{definition}[{Matroids}]
	An independence system $(N, \calI)$ is called a \textit{matroid} if for any $S, T \in \calI$ with $|S| < |T|$, there exists $v \in T \setminus S$ such that $S \cup \{v\} \in \calI$.
\end{definition}

\begin{definition}[{$p$-Matroid intersection}]
	An independence system $(N, \calI)$ is a \textit{$p$-matroid intersection} if there exist $p$ matroids $(N, \calI_1), \cdots, (N, \calI_p)$ such that $\calI = \bigcap_{i=1}^p \calI_i$.
\end{definition}

\begin{definition}[{$p$-Exchange systems~\citep{FNSW11}}]
	An independence system $(N, \calI)$ is a \textit{$p$-exchange system} if for any $S, T \in \calI$, there exists a map $\varphi \colon (T \setminus S) \to 2^{S \setminus T}$ such that (a) for any $v \in T \setminus S$, it holds that $|\varphi(v)| \le p$, (b) each $v \in S \setminus T$ appears in $(\varphi(v))_{v \in T \setminus S}$ at most $p$ times, and (c) for any $X \subseteq T \setminus S$, it holds that $(S \setminus \bigcup_{v \in X} \varphi(v)) \cup X \in \calI$.
\end{definition}

\paragraph{Sparse optimization.}
Sparse optimization is the problem of finding a sparse solution that maximizes a continuously differentiable function $u \colon \bbR^n \to \bbR$.
Assume we have an access to a zeroth and first-order oracle that returns the value of $u(\bfw)$ and gradient $\nabla u(\bfw)$ given $\bfw \in \bbR^n$.
To define the approximation ratio properly, we need to assume $u(\bfzero) \ge 0$, but we can normalize any function $u' \colon \bbR^n \to \bbR$ by setting $u(\bfw) \coloneqq u'(\bfw) - u'(\bfzero)$.
Let $N = [n]$ be the set of all variables and $\calI \subseteq 2^N$ a family of feasible supports.
We can write a sparse optimization problem with structured constraints as
\begin{equation}
	\text{Maximize} \quad u(\bfw) \qquad \text{subject to} \quad \supp( \bfw ) \in \calI,
\end{equation}
where $\supp(\bfw)$ represents the set of non-zero elements of $\bfw$, that is, $\supp(\bfw) = \{ i \in N \mid \bfw_i \neq 0 \}$.
We define $\| \bfw \|_0 = |\supp(\bfw)|$.

We assume restricted strong concavity and restricted smoothness of the objective function $u$, which are defined as follows.

\begin{definition}[{Restricted strong concavity and restricted smoothness~\citep{NRWY12,JTK14}}]\label{def:restricted}
	Let $\Omega$ be a subset of $\bbR^d \times \bbR^d$ and $u \colon \bbR^d \to \bbR$ be a continuously differentiable function.
	We say that $u$ is \emph{restricted strongly concave} with parameter $m_\Omega$ and \emph{restricted smooth} with parameter $M_\Omega$ on domain $\Omega$ if
	\begin{align*}
			- \frac{m_\Omega}{2} \| \bfy - \bfx \|_2^2
			&\ge u(\bfy) - u(\bfx) - \langle \nabla u (\bfx), \bfy - \bfx \rangle\\
			&\ge - \frac{M_\Omega}{2} \| \bfy - \bfx \|^2_2
	\end{align*}
    for all $(\bfx, \bfy) \in \Omega$.
\end{definition}

Let $\Omega_{s} = \{ (\bfx, \bfy) \in \bbR^n \times \bbR^n \mid \| \bfx \|_0 \le s, \| \bfy \|_0 \le s, \| \bfx - \bfy \|_0 \le s \}$ and $\Omega_{s,t} = \{ (\bfx, \bfy) \in \bbR^n \times \bbR^n \mid \| \bfx \|_0 \le s , ~ \| \bfy \|_0 \le s, ~ \| \bfx - \bfy \|_0 \le t \}$.
Let $m_s$ be restricted strong concavity parameter on $\Omega_s$ and $M_{s,t}$ the restricted smoothness parameter on $\Omega_{s,t}$ for any positive integer $s,t \in \bbZ_{> 0}$.
Due to the restricted strong concavity of $u$, $\argmax_{\supp(\bfw) \subseteq X} u(\bfw)$ is uniquely determined.
We denote this maximizer by $\bfw^{(X)}$.

By introducing a set function $f \colon 2^{N} \to \bbR_{\ge 0}$ defined as
\begin{equation*}
	f(X) = \max_{\supp(\bfw) \subseteq X} u(\bfw),
\end{equation*}
we can regard the sparse optimization problem as a set function optimization problem \eqref{eq:set-max}.

\paragraph{Notations.}
Vectors are denoted by bold lower-case letters (e.g. $\bfx$ and $\bfy$) and matrices are denoted by bold upper-case letters (e.g. $\bfA$ and $\bfB$).
Sets are denoted by upper-case letters (e.g. $X$ and $Y$).
For $X \subseteq N$ and $a \in N$, we define $X + a \coloneqq X \cup \{a\}$ and $X - a \coloneqq X \setminus \{a\}$.
We define the symmetric difference by $X \triangle Y \coloneqq (X \setminus Y) \cup (Y \setminus X)$.

\section{Localizability of Set Functions}\label{sec:local-approximate}
In this section, we define a property of set functions, which we call \textit{localizability}.
Since the general version of the definition of localizability is complicated, we first introduce a simplified definition as a warm-up, and then state the general definition.

Intuitively, localizability measures how much small modifications of a solution increase the objective value.
Localizability is defined as a property that the sum of the increases yielded by small modifications is no less than the increase yielded by a large modification.
\begin{definition}[{Localizability (simplified version)}]\label{def:localizability}
	Let $f \colon 2^N \to \bbR_{\ge 0}$ be a non-negative monotone set function.
	For some $\alpha, \beta \in \bbR_{\ge 0}$, we say $f$ is \textit{$(\alpha, \beta)$-localizable with size $s$} if for arbitrary subsets $X, X^* \subseteq N$ of size $s$ and bijection $\phi \colon X \setminus X^* \to X^* \setminus X$, we have
	\begin{equation*}
		\sum_{x \in X \setminus X^*} \left\{ f(X - x + \phi(x)) - f(X) \right\} \ge \alpha f(X^*) - \beta f(X).
	\end{equation*}
\end{definition}

This property is sufficient to provide an approximation guarantee of local search algorithms for matroid constraints.
However, we need a generalized version of localizability that considers exchanges of multiple elements to deal with more complicated constraints.
\begin{definition}[{Localizability}]\label{def:localizability-general}
	Let $f \colon 2^N \to \bbR_{\ge 0}$ be a non-negative monotone set function.
	For some $\alpha, \beta_1, \beta_2 \in \bbR_{\ge 0}$, we say $f$ is \textit{$(\alpha, \beta_1, \beta_2)$-localizable with size $s$ and exchange size $t$} if for arbitrary subsets $X, X^* \subseteq N$ of size at most $s$ and any collection $\calP$ of subsets of $X \triangle X^*$ such that $|P| \le t$ for each $P \in \calP$, we have
	\begin{equation*}
		\sum_{P \in \calP} \left\{ f(X \triangle P) - f(X) \right\} \ge \alpha k f(X^*) - ( \beta_1 \ell + \beta_2 k) f(X),
	\end{equation*}
	where $k$ and $\ell$ are positive integers such that each element in $X^* \setminus X$ appears at least $k$ times in $\calP$ and each element in $X \setminus X^*$ appears at most $\ell$ times in $\calP$.
\end{definition}

If we consider the case when $k=1$ and $\ell=1$, this definition coincides with the simplified version with $\beta = \beta_1 + \beta_2$.
In existing studies on submodular maximization, \citet{LSV10} and \citet{FNSW11} utilized this property of linear functions and non-negative monotone submodular functions to prove the approximation bounds for local search algorithms.
\begin{proposition}[{Proved in the proof of Lemma 3.1 of \citet{LSV10}}]
	Any linear function is $(1, 1, 0)$-localizable and any non-negative monotone submodular function is $(1, 1, 1)$-localizable with any size and any exchange size.
\end{proposition}

In the following proposition, we show that the set function derived from sparse optimization satisfies localizability under the restricted strong concavity and restricted smoothness assumption.
\begin{proposition}\label{lem:feature-exchange}
	Suppose $u \colon 2^N \to \bbR$ is a continuously differentiable function with $u(\bfzero) \ge 0$.
	Let $s,t \in \bbZ_{\ge 0}$ be arbitrary integers.
	Assume $u$ is restricted strong concave on $\Omega_{2s}$ and restricted smooth on $\Omega_{s,t}$.
	If $f \colon 2^N \to \bbR$ is a set function defined as $f(X) = \max_{\supp(\bfw) \subseteq X} u(\bfw)$, then $f$ is $\displaystyle \left( \frac{m_{2s}}{M_{s,t}}, \frac{M_{s,t}}{m_{2s}}, 0 \right)$-localizable with size $s$ and exchange size $t$.
\end{proposition}

\section{Local Search Algorithms}\label{sec:local-algorithms}
In this section, we describe our proposed local search algorithms for a matroid constraint, a $p$-matroid intersection constraint, and a $p$-exchange system constraint, and provide approximation ratio bounds in terms of localizability.

\subsection{Algorithms for a Matroid Constraint}\label{sec:local-matroid}
Here, we describe our proposed algorithm for a matroid constraint.
The algorithm starts with an initial solution, which is any base of the given matroid.
The main procedure of the algorithm is to repeatedly improve the solution by replacing an element in the solution with another element.
At each iteration, the algorithm seeks a pair of an element $x \in X$ and another element $x' \in N \setminus X$ that maximizes $f(X - x + x')$ while keeping the feasibility, which requires $\rmO(sn)$ oracle calls.
The detailed description of the algorithms is given in \Cref{alg:matroid-anytime}.

\begin{algorithm}[t]
	\caption{Local search algorithms for a matroid constraint}\label{alg:matroid-anytime}
	\begin{algorithmic}[1]
		\STATE Let $X \gets \emptyset$.
		\STATE Add arbitrary elements to $X$ until $X$ is maximal in $\calI$.
		\FOR{$i = 1,\cdots,T$}
			\STATE Find the pair of $x \in X$ and $x' \in N \setminus X$ such that $\displaystyle (x, x') \in \argmax \{ f(X - x + x') \mid X - x + x' \in \calI \}$
			\IF{$f(X - x + x') - f(X) > 0$}
				\STATE Update the solution $X \gets X - x + x'$.
			\ELSE
				\STATE \textbf{return} $X$.
			\ENDIF
		\ENDFOR
		\STATE \textbf{return} $X$.
	\end{algorithmic}
\end{algorithm}

We can provide an approximation guarantee in terms of localizability of the objective function as follows.
\begin{theorem}\label{thm:matroid-anytime}
	Suppose $\calI$ is the independence set family of a matroid and $s = \max \{ |X| \mid X \in \calI \}$.
	Assume the objective function $f$ is non-negative, monotone, and $(\alpha, \beta_1, \beta_2)$-localizable with size $s$ and exchange size $2$.
	If $X$ is the solution obtained by executing $T$ iterations of \Cref{alg:matroid-anytime} and $X^*$ is an optimal solution, then we have
	\begin{equation*}
		f(X) \ge \frac{\alpha}{\beta_1 + \beta_2} \left( 1 - \exp\left( - \frac{(\beta_1 + \beta_2) T}{s} \right) \right) f(X^*).
	\end{equation*}
	If $X$ is the output returned by \Cref{alg:matroid-anytime} when it stops by finding no pair to improve the solution, then we have
	\begin{equation*}
		f(X) \ge \frac{\alpha}{\beta_1 + \beta_2} f(X^*).
	\end{equation*}
\end{theorem}

\subsection{Algorithms for $p$-Matroid Intersection and $p$-Exchange System Constraints}\label{sec:local-intersection}
In this section, we consider two more general constraints, $p$-matroid intersection and $p$-exchange system constraints with $p \ge 2$.
The proposed algorithms for these two constraints can be described as almost the same procedure by using the different definitions of $q$-reachability as follows.

\begin{definition}[{$q$-Reachability for $p$-matroid intersection~\citep{LSV10}}]\label{def:reachability-intersection}
	Let $\calI \subseteq 2^N$ be a $p$-matroid intersection.
	A feasible solution $T \in \calI$ is $q$-reachable from $S \in \calI$ if $|T \setminus S| \le 2q$ and $|S \setminus T| \le 2pq$.
\end{definition}

\begin{definition}[{$q$-Reachability for $p$-exchange systems~\citep{FNSW11}}]\label{def:reachability-system}
	Let $\calI \subseteq 2^N$ be a $p$-exchange system.
	A feasible solution $T \in \calI$ is $q$-reachable from $S \in \calI$ if $|T \setminus S| \le q$ and $|S \setminus T| \le pq - q + 1$.
\end{definition}

We denote by $\calF_q(X)$ the set of all $q$-reachable sets from $X$ that is determined by each definition of $q$-reachability for $p$-matroid intersection or $p$-exchange systems.

First, we must decide parameter $q \in \bbZ_{\ge 1}$ that determines the neighborhood to be searched at each iteration.
When we select larger $q$, we search larger solution space for improvement at each iteration; thus, we can obtain a better bound on its approximation ratio, while the number of oracle calls $n^{\rmO(q)}$ becomes larger as well.
The initial solution of the proposed algorithms is any feasible solution.
Then the algorithms repeatedly replace the solution with a $q$-reachable solution that increases the objective value the most.
The detailed description of this local search algorithm is given in \Cref{alg:system-anytime}.

\begin{algorithm}[t]
	\caption{Local search algorithms for a $p$-matroid intersection or $p$-exchange system constraint ($p \ge 2$)}\label{alg:system-anytime}
	\begin{algorithmic}[1]
		\STATE Let $X \gets \emptyset$.
		\FOR{$i = 1,\cdots,T$}
			\STATE Find $\displaystyle X' \in \argmax_{X' \in \calF_q(X)} f(X')$.
			\IF{$f(X') - f(X) > 0$}
				\STATE Update the solution $X \gets X'$.
			\ELSE
				\STATE \textbf{return} $X$.
			\ENDIF
		\ENDFOR
		\STATE \textbf{return} $X$.
	\end{algorithmic}
\end{algorithm}

We can provide an approximation ratio bound under the assumption of localizability of the objective function as follows.
\begin{theorem}\label{thm:system-anytime}
	Suppose $\calI$ is the independence set family of a $p$-matroid intersection or $p$-exchange system and $s = \max \{ |X| \mid X \in \calI \}$.
	Let $t = 2p(q+1)$ for the $p$-matroid intersection case and $t = pq+1$ for the $p$-exchange system case.
	Assume the objective function $f$ is non-negative, monotone, and $(\alpha, \beta_1, \beta_2)$-localizable with size $s$ and exchange size $t$.
	If $X$ is the output obtained by executing $T$ iterations of \Cref{alg:system-anytime} with parameter $q$ and $X^*$ is an optimal solution, then the approximation ratio is lower-bounded by
	\begin{equation*}
		\frac{\alpha \left( 1 - \exp\left( \frac{(\beta_1 (p - 1 + 1/q) + \beta_2 )  T}{s}\right) \right)}{\beta_1 (p - 1 + 1/q) + \beta_2}.
	\end{equation*}
	If $X$ is the output returned by \Cref{alg:system-anytime} when it stops by finding no better $q$-reachable solution, then we have
	\begin{equation*}
		f(X) \ge \frac{\alpha}{\beta_1 (p - 1 + 1/q) + \beta_2} f(X^*).
	\end{equation*}
\end{theorem}

\section{Acceleration for Sparse Optimization}\label{sec:local-acceleration}
We consider two accelerated variants of the proposed local search algorithms in the case of sparse optimization.
To distinguish the original one from the accelerated variants, we call \Cref{alg:matroid-anytime} and \Cref{alg:system-anytime} \textit{oblivious local search algorithms}.

\subsection{Acceleration for a Matroid Constraint}

The oblivious version computes the value of $f(X - x + x')$ for $\rmO(sn)$ pairs of $(x, x')$ at each iteration.
We can reduce the computational cost by utilizing the structure of sparse optimization.

The first variant is the \textit{semi-oblivious} local search algorithm.
For each element $x' \in N \setminus X$ to be added, it computes the value of $f(X - x + x')$ only for $x \in X$ with the smallest $(\bfw^{(X)})^2_x$ among those satisfying $X - x + x' \in \calI$.
Thus, we can reduce the number of times we compute the value of $f(X - x + x')$ from $\rmO(sn)$ to $\rmO(n)$.

The second variant is the \textit{non-oblivious} local search algorithm.
It uses the value of
\begin{equation*}
	\frac{1}{2M_{s,2}} \left( \nabla u(\bfw^{(X)}) \right)^2_{x'} - \frac{M_{s,2}}{2} \left(\bfw^{(X)}\right)_{x}^2
\end{equation*}
in place of the increase of the objective function $f(X - x + x') - f(X)$.
We need to evaluate $\nabla u (\bfw^{(X)})$ and $\bfw^{(X)}$ at the beginning of each iteration, but it is not necessary to compute the value of $f(X - x + x')$.

The detailed description of these algorithms are given in \Cref{alg:matroid-anytime-full} in \Cref{sec:full-pseudocodes}.

\begin{theorem}\label{thm:matroid-sparse}
	Suppose $f(X) = \max_{\supp(\bfw) \subseteq X} u(\bfw)$ and $\calI$ is the independence set family of a matroid.
	If $X$ is the solution obtained by executing $T$ iterations of the semi-oblivious or non-oblivious local search algorithms and $X^*$ is an optimal solution, then we have
	\begin{equation*}
		f(X) \ge \frac{m_{2s}^2}{M_{s,2}^2} \left( 1 - \exp\left( - \frac{M_{s,2} T}{s m_{2s}} \right) \right) f(X^*),
	\end{equation*}
	where $s = \max \{|X| \colon X \in \calI \}$.
	If $X$ is the output returned when the algorithm stops by finding no pair to improve the solution, then we have
	\begin{equation*}
		f(X) \ge \frac{m_{2s}^2}{M_{s,2}^2} f(X^*).
	\end{equation*}
\end{theorem}

\subsection{Acceleration for $p$-Matroid Intersection and $p$-Exchange System Constraints}

Similarly to the case of matroid constraints, we can develop the semi-oblivious and non-oblivious local search algorithms for $p$-matroid intersection and $p$-exchange system constraints.
The semi-oblivious variant only checks $X' \in \calF_q(X)$ that minimizes $\left\| \left( \bfw^{(X)} \right)_{X \setminus X'} \right\|^2$ among $X'' \in \calF_q(X)$ such that $X'' \setminus X = X' \setminus X$.
The non-oblivious version selects the solution $X' \in \calF_q(X)$ that maximizes
\begin{equation*}
	\frac{1}{2M_{s,t}} \left\| \left( \nabla u(\bfw^{(X)}) \right)_{X' \setminus X} \right\|^2 - \frac{M_{s,t}}{2} \left\| \left(\bfw^{(X)}\right)_{X \setminus X'} \right\|^2.
\end{equation*}
The detailed description of these algorithms are given in \Cref{alg:system-anytime-full} in \Cref{sec:full-pseudocodes}.
While the oblivious local search algorithm requires $O(n^q)$ times of the evaluation of $f$ for finding the most suitable exchange at each iteration in general, the non-oblivious local search reduces it to a linear function maximization problem.
In several cases such as a partition matroid constraint or a matching constraint, we can find the most suitable exchange in time polynomial in $n$ and $q$ by using standard techniques of combinatorial optimization.
We can provide the same approximation guarantees for these accelerated variants as the oblivious variant.

\begin{theorem}\label{thm:system-sparse}
	Suppose $f(X) = \max_{\supp(\bfw) \subseteq X} u(\bfw)$ and $\calI$ is the independence set family of a $p$-matroid intersection or $p$-exchange system.
	Let $t = 2p(q+1)$ for the $p$-matroid intersection case and $t = pq+1$ for the $p$-exchange system case.
	If $X$ is the output obtained by executing $T$ iterations of the semi-oblivious or non-oblivious local search algorithms with parameter $q$ and $X^*$ is an optimal solution, then its approximation ratio is lower-bounded by
	\begin{equation*}
		\frac{1}{p - 1 + 1/q} \frac{m_{2s}^2}{M_{s,t}^2} \left( 1 - \exp\left( - \frac{(p - 1 + 1/q) M_{s,t} T}{s m_{2s}} \right) \right),
	\end{equation*}
	where $s = \max \{|X| \colon X \in \calI \}$.
	If $X$ is the output returned when the algorithm stops by finding no better $q$-reachable solution, then we have
	\begin{equation*}
		f(X) \ge \frac{1}{p - 1 + 1/q} \frac{m_{2s}^2}{M_{s,t}^2} f(X^*).
	\end{equation*}
\end{theorem}

\begin{remark}
	We also develop another version of our local search algorithms that increases the objective value at a predetermined rate, which is described in \Cref{sec:local-geometric}.
\end{remark}
\begin{remark}
	The parameter $M_{s,t}$ used in the non-oblivious variant can be replaced with an upper bound on $M_{s,t}$, which leads to the approximation ratio bounds whose $M_{s,t}$ is also replaced with the upper bound.
\end{remark}

\section{Applications}\label{sec:local-applications}
In this section, we provide two applications of our framework: feature selection for sparse regression and structure learning of graphical models.

\subsection{Feature Selection for Sparse Regression}
In sparse regression, given a design matrix $\bfA \in \bbR^{n \times d}$ and a response vector $\bfy$, we aim to find a sparse vector $\bfw \in \bbR^n$ that optimizes some criterion.
We can formulate this problem as a sparse optimization problem of maximizing $u(\bfw)$ subject to $|\supp(\bfw)| \le s$, where $u \colon \bbR^{n} \to \bbR$ is the criterion determined by $\bfA$ and $\bfy$.
\citet{Das2011} devised approximation algorithms in the case where $u$ is the squared multiple correlation $R^2$, i.e., $u(\bfw) \coloneqq 1 - \| \bfy - \bfA \bfw \|_2^2 / \| \bfy \|_2^2$, and \citet{Elenberg18} extended their results to general objectives with restricted strong concavity and restricted smoothness.

Here we consider sparse regression with \textit{structured constraints}.
In practical scenarios, we often have prior knowledge of relationships among features and can improve the quality of the estimation by incorporating it into structured constraints \citep{BCDH10,Huang2009}.
We formulate sparse regression with structured constraints as the problem of maximizing $u(\bfw)$ subject to $\supp(\bfw) \in \calI$, where $\calI$ is the set family of feasible supports.

An advantage of our local search framework is its applicability to a broad class of structured constraints, including matroid constraints.
For example, the following constraint is a special case of matroid constraints.
Suppose the set of features are partitioned into several categories.
Due to a balance among categories, it is often the case that we should select almost the equal number of features from each category.
Such a constraint can be expressed by using a partition matroid.
Partition matroid constraints were used for multi-level subsampling by \citet{BLSGB16} and detecting splice sites in precursor messenger RNAs by \citet{CFK18}.
If there are multiple matroid constraints, we can formulate them as a $p$-matroid intersection constraint.
To our knowledge, our proposed algorithms are the first to cope with multiple matroid constraints.

\subsection{Structure Learning of Graphical Models}
Undirected graphical models, or Markov random fields, express the conditional dependence relationships among random variables.
We consider the problem of estimating the graph structure of an undirected graphical model given samples generated from this probability distribution.
The goal of this problem is to restore the set of edges, that is, the set of all conditionally dependent pairs.
To obtain a more interpretable graphical model, we often impose a sparsity constraint on the set of edges.
This task can be formulated as a sparse optimization problem.

While most existing methods solve the neighborhood estimation problem separately for each vertex under a sparsity constraint \citep{JJR11,KM17}, our framework provides an optimization method that handles the sparsity constraints for all vertices simultaneously.
Suppose we aim to maximize some likelihood function (e.g., pseudo-log-likelihood \citep{Besag75}) under the sparsity constraint on each vertex, i.e., the degree of each vertex is at most $b$, where $b \in \bbZ_{\ge 0}$ is the maximum degree.
This degree constraint is called a $b$-matching constraint, which is a special case of $2$-exchange system.
Hence, we can apply our local search algorithms to this problem.

\section{Experiments}\label{sec:local-experiments}
In this section, we conduct experiments on two applications: sparse regression and structure learning of graphical models.
All the algorithms are implemented in Python 3.6.
We conduct the experiments in a machine with Intel Xeon E3-1225 V2 (3.20 GHz and 4 cores) and 16 GB RAM.

\begin{figure*}[t]
\centering
\subfigure[regression, time, $n = 200$]{
	\includegraphics[width=0.3\textwidth]{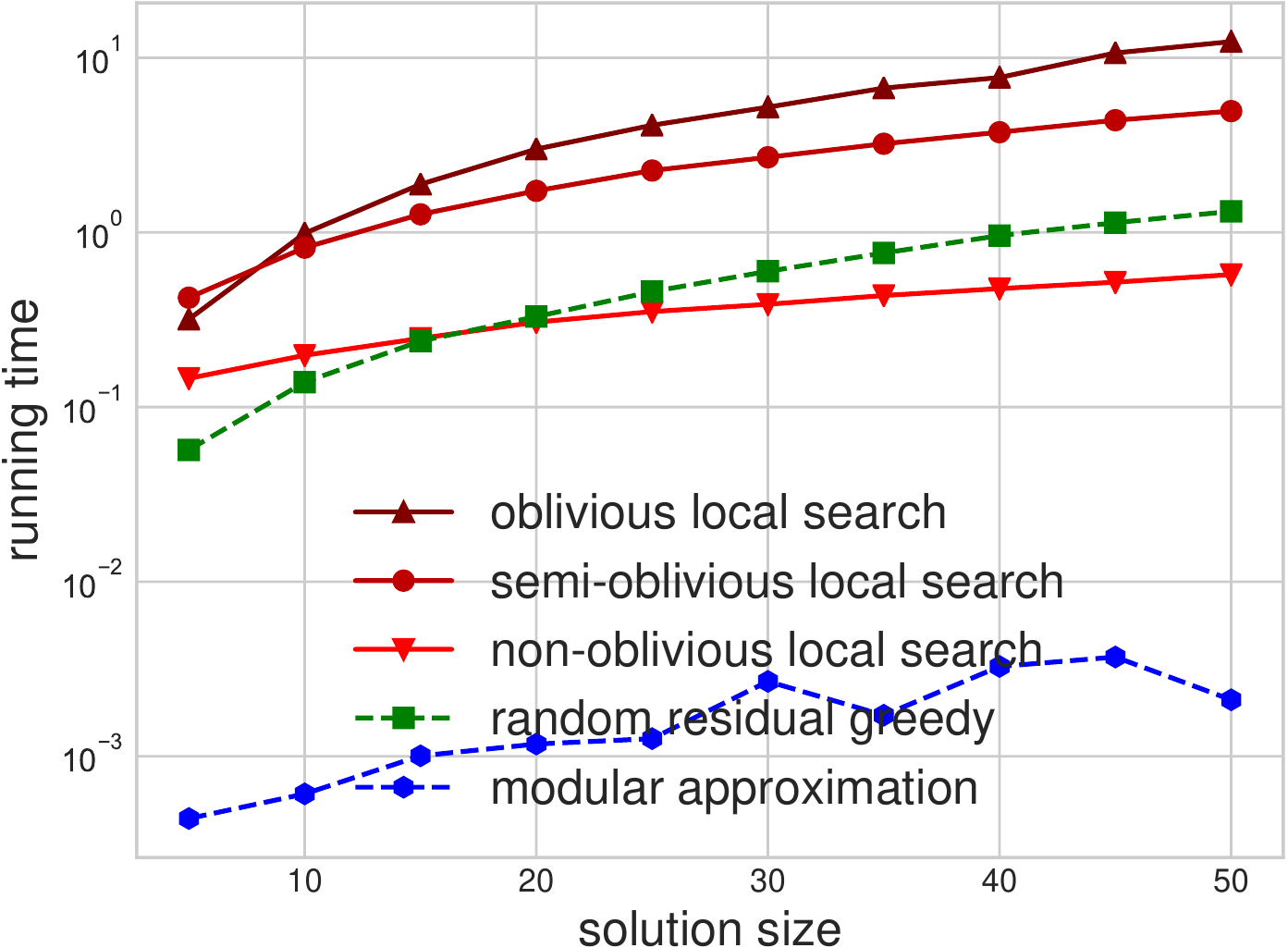}\label{fig:regression_n200_time}
}
\subfigure[regression, objective, $n = 200$]{
	\includegraphics[width=0.3\textwidth]{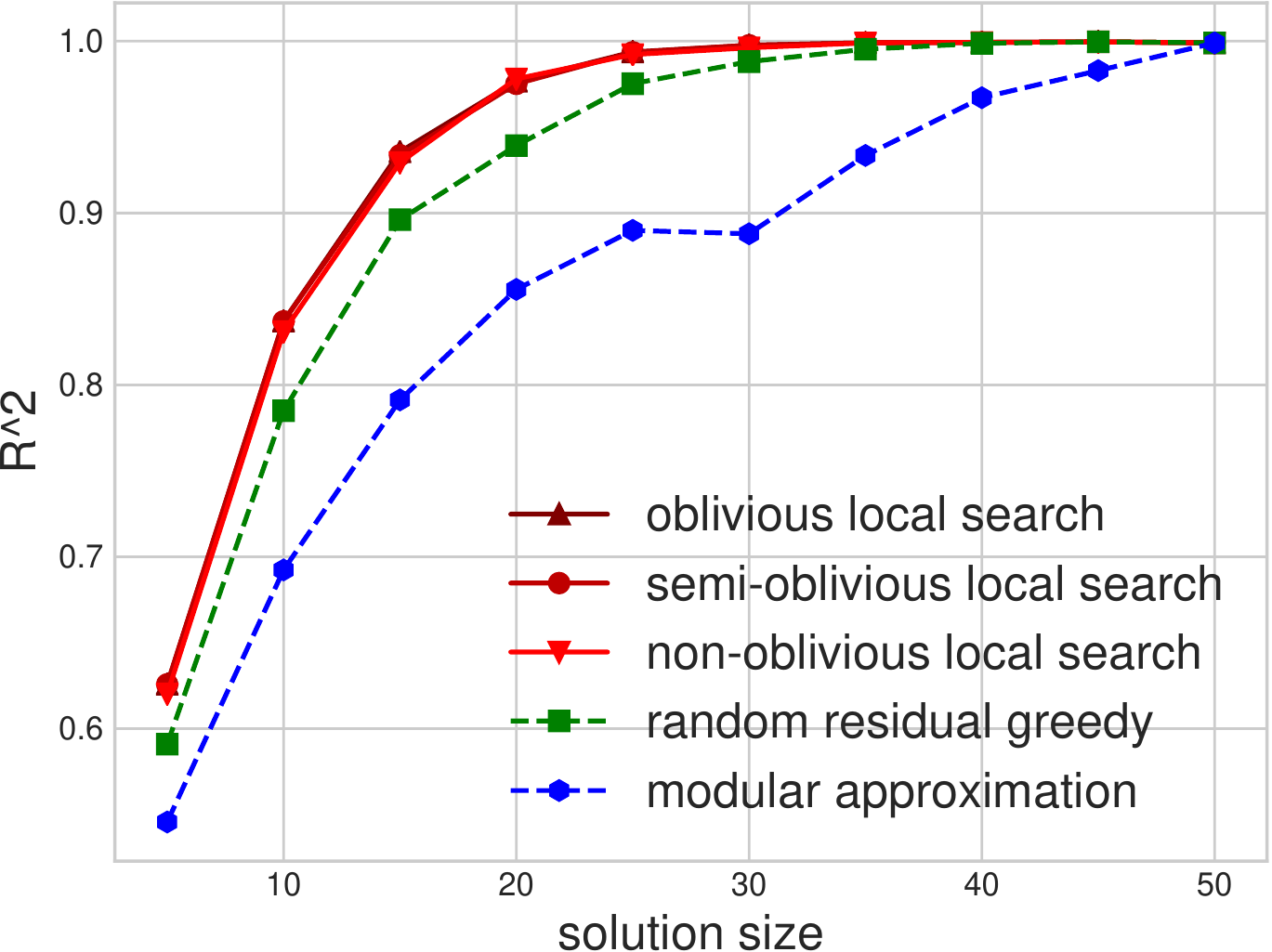}\label{fig:regression_n200}
}
\subfigure[regression, objective, $n = 1000$]{
	\includegraphics[width=0.3\textwidth]{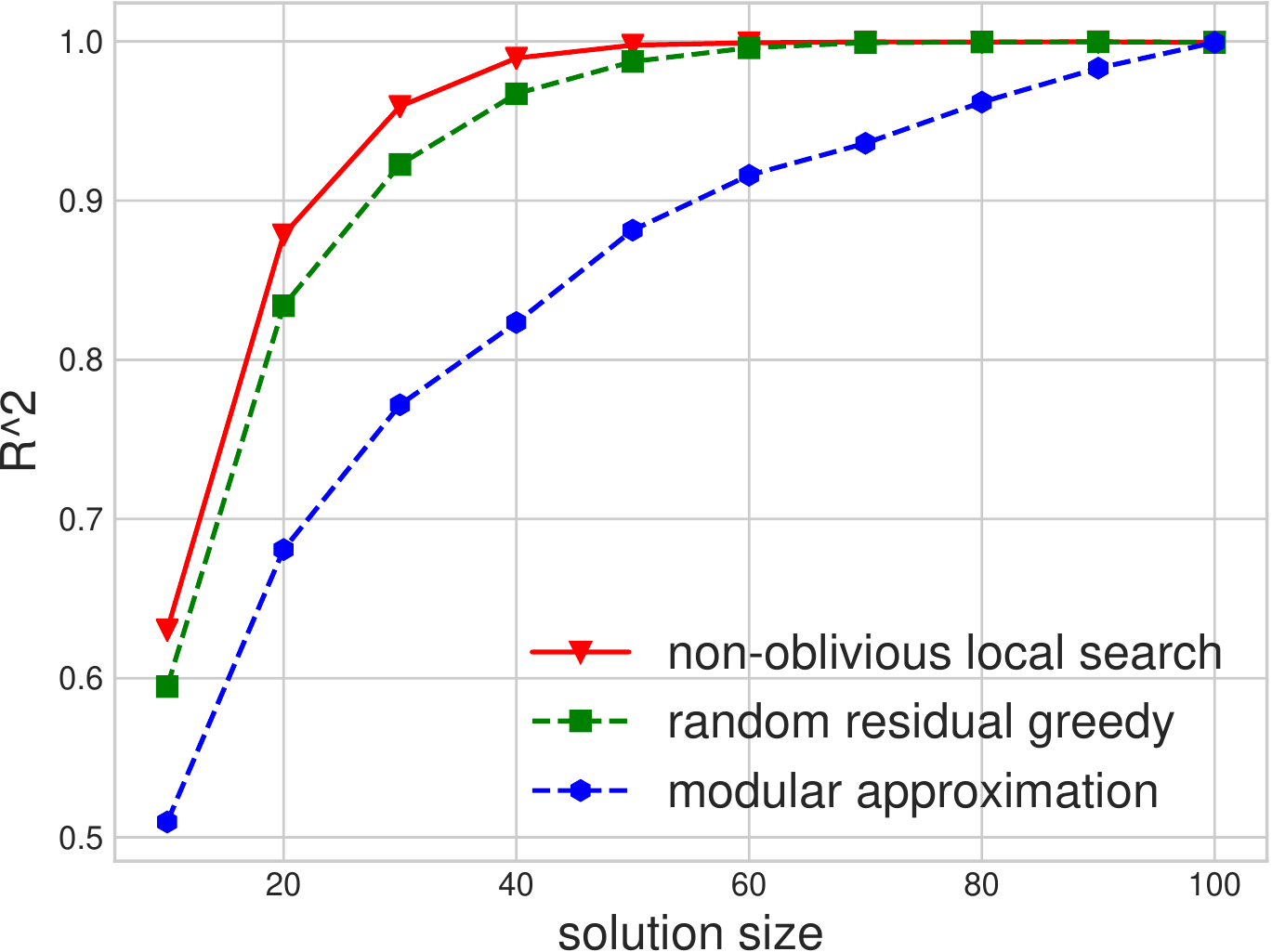}\label{fig:regression_n1000}
}
\subfigure[graphical, time, $|V| = 10$]{
	\includegraphics[width=0.3\textwidth]{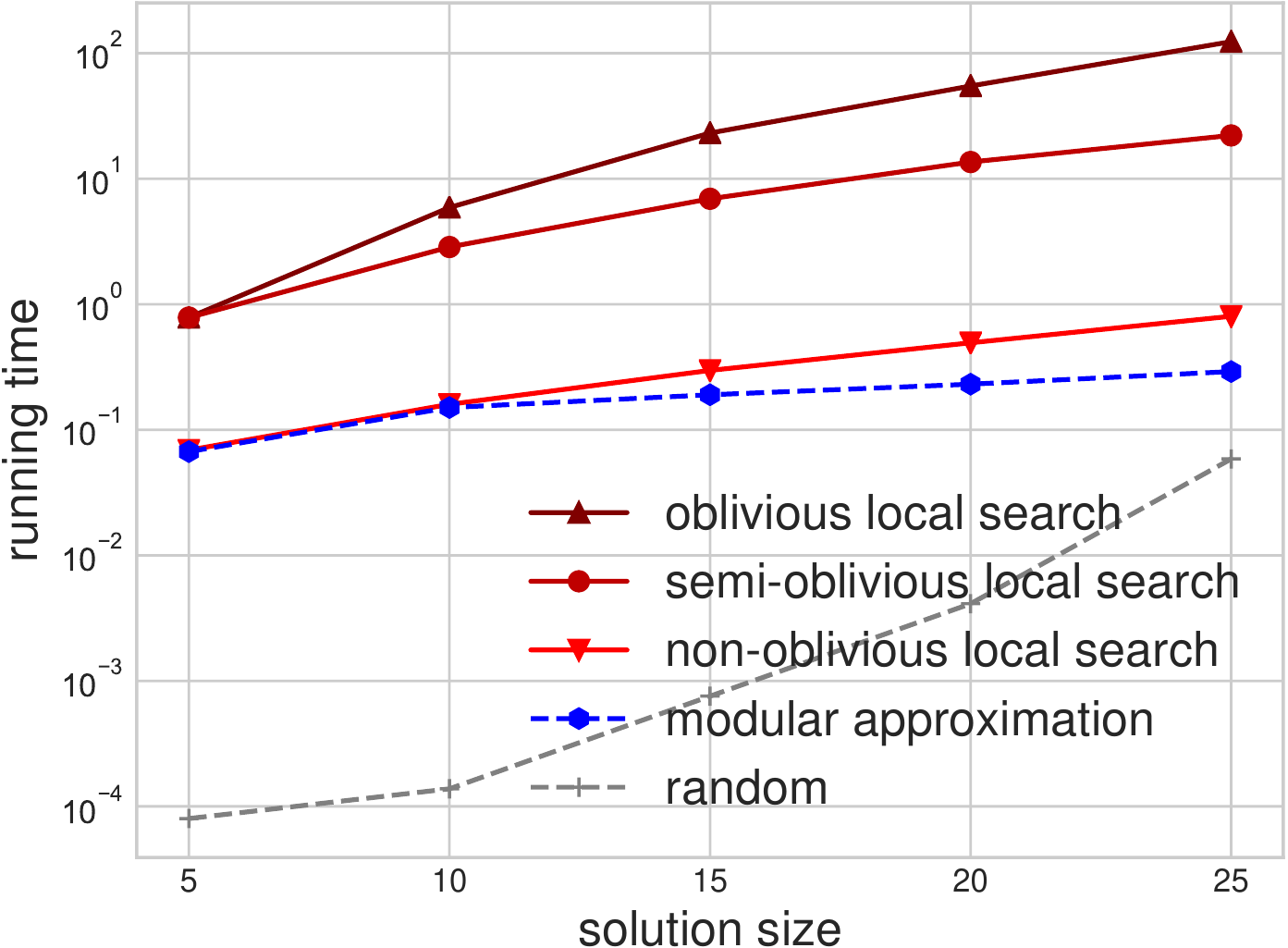}\label{fig:graphical_n10_time}
}
\subfigure[graphical, objective, $|V| = 10$]{
	\includegraphics[width=0.3\textwidth]{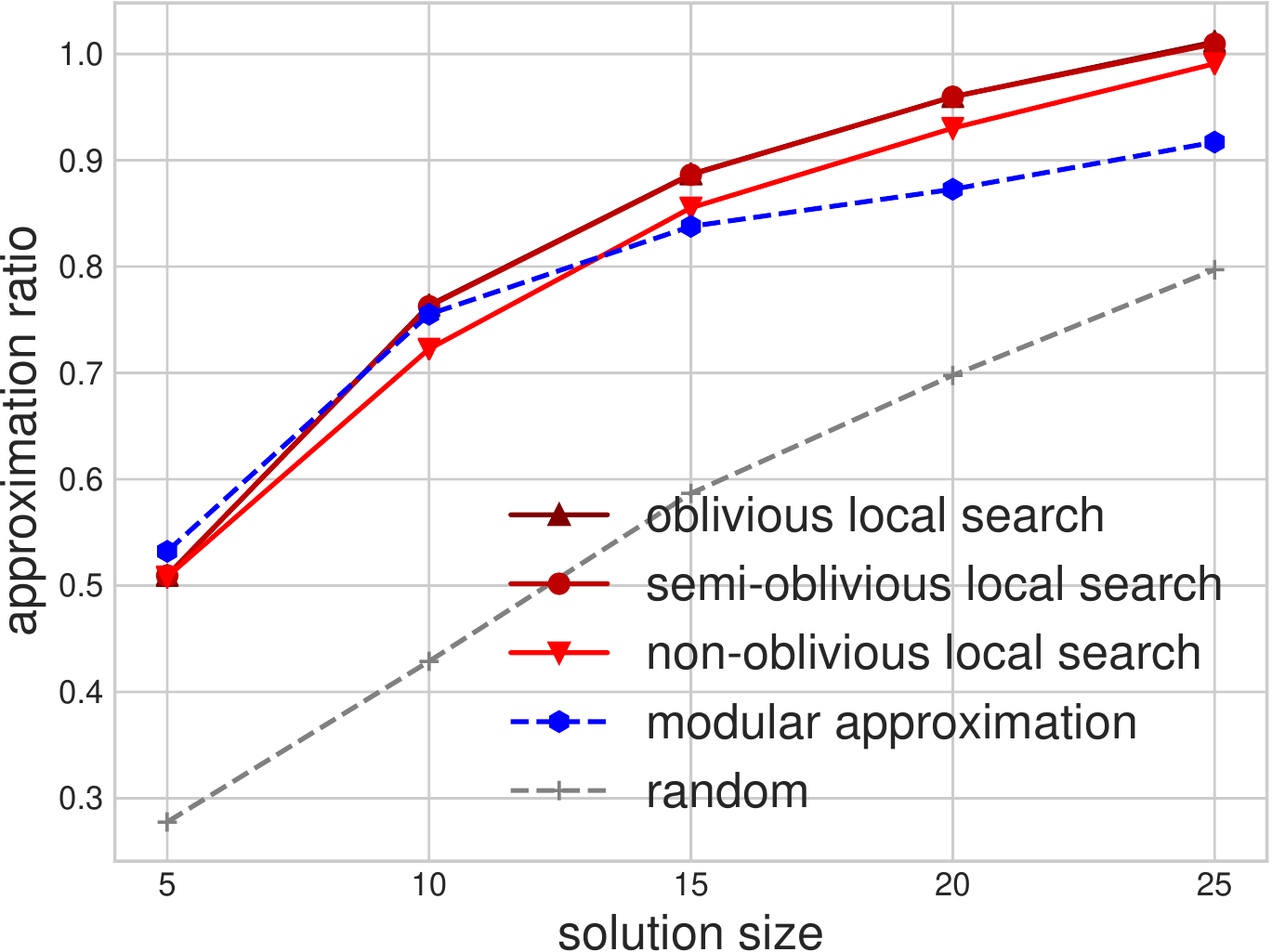}\label{fig:graphical_n10}
}
\subfigure[graphical, objective, $|V| = 20$]{
	\includegraphics[width=0.3\textwidth]{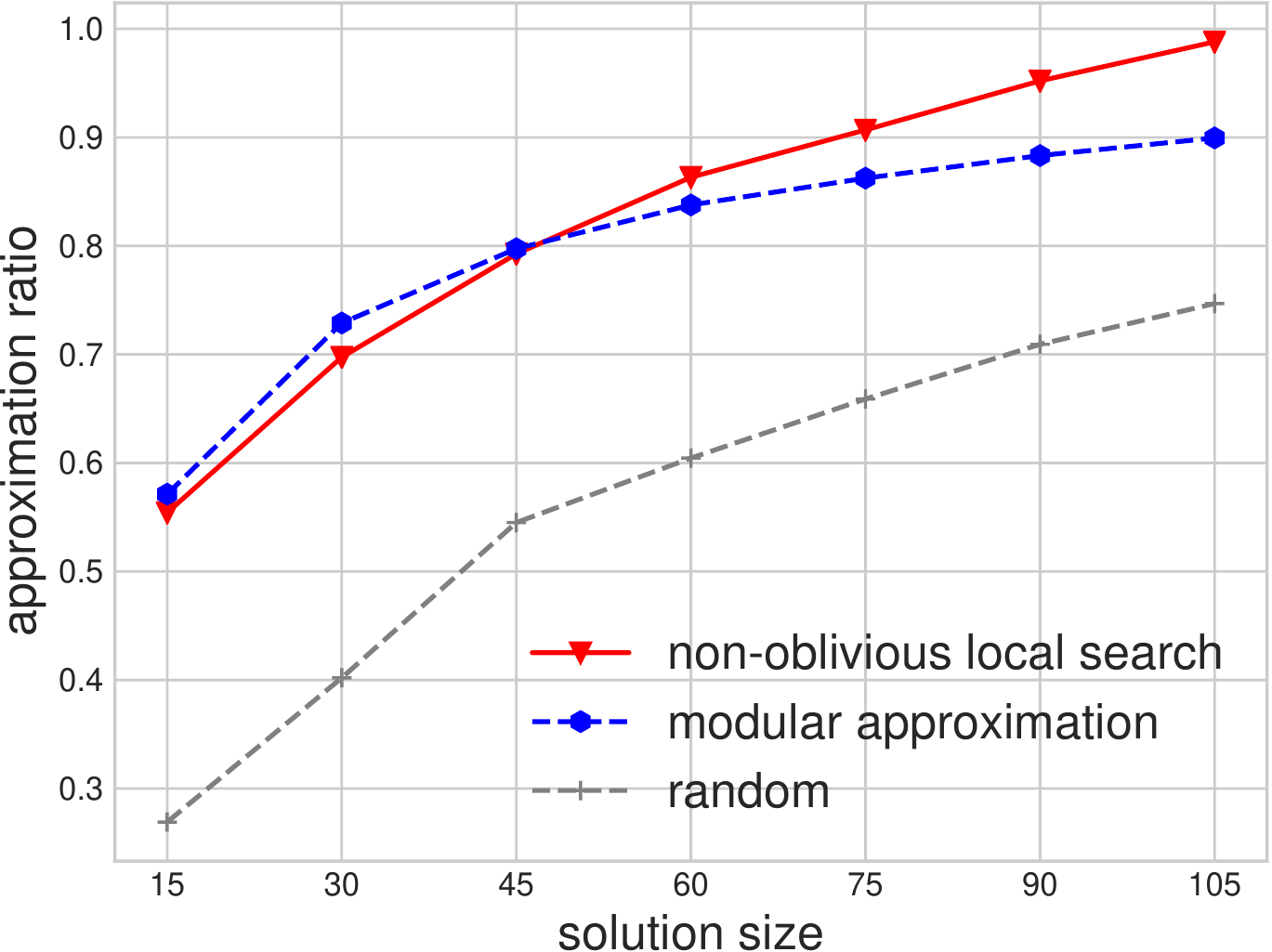}\label{fig:graphical_n20}
}
\caption{
	The experimental results on sparse linear regression under a partition matroid constraint (\ref{fig:regression_n200_time}, \ref{fig:regression_n200}, and \ref{fig:regression_n1000}) and structure learning of graphical models under a $b$-matching constraint (\ref{fig:graphical_n10_time}, \ref{fig:graphical_n10}, and \ref{fig:graphical_n20}).
	\ref{fig:regression_n200_time} shows the running time in the case where $n = 100$.
	\ref{fig:regression_n200} and \ref{fig:regression_n1000} show the objective value ($R^2$) in the case where $n = 200$ and $n = 1000$, respectively.
	\ref{fig:graphical_n10_time} shows the running time in the case where $|V| = 10$.
	\ref{fig:graphical_n10} and \ref{fig:graphical_n20} show the ratio between the objective achieved by the algorithms and the optimal objective value in the case where $|V| = 10$ and $|V| = 20$, respectively.
	}
	\label{fig:local_search_regression}
\end{figure*}

\subsection{Experiments on Sparse Regression}

\paragraph{Datasets.}
We generate synthetic datasets with a partition matroid constraint.
First, we determine the design matrix $\bfA \in \bbR^{n \times d}$ by generating each of its entries according to the uniform distribution on $[0, 1]$.
Then we normalize each of columns to ensure that the mean will be $0$ and the standard deviation will be $1$.
Suppose the set of all features are partitioned into $n_c$ equal-size categories.
We randomly select a sparse subset $S^*$ by selecting $n_p$ parameters from each category.
The response vector is determined by $\bfy = \bfA_{S^*} \bfw + \epsilon$, where $\bfw$ is a random vector generated from the standard normal distribution $\calN(0, 1)$ and $\epsilon$ is a noise vector whose each element is generated from $\calN(0, 0.2)$.
We consider two settings with different parameters.
We set $(n, d, n_c, n_p) = (200, 50, 5, 5)$ in one setting and $(n, d, n_c, n_p) = (1000, 100, 10, 10)$ in the other setting.
We use $R^2$ as the objective function to be maximized.
For each parameter, we conduct $10$ trials and plot the average.

\paragraph{Methods.}
We implement the oblivious, semi-oblivious, and non-oblivious local search algorithms.
As benchmarks, we implement the random residual greedy algorithm \citep{CFK18} and modular approximation.
We apply these methods to a partition matroid constraint with capacity $n'_p$ for each $n'_p \in \{1,\cdots,10\}$.

\paragraph{Results.}
First, we compare the proposed methods in the case of $n = 200$ (\Cref{fig:regression_n200_time} and \Cref{fig:regression_n200}).
We can observe that the non-oblivious variant is approximately $10$ times faster than the oblivious variant, while achieving an objective value comparable to that of the oblivious variant.
In comparison to the random residual greedy algorithm, the non-oblivious variant achieves a higher objective value in a similar running time.
Modular approximation is considerably faster than the other methods, but the quality of its solution is poor.
Next, we conduct experiments on larger datasets with $n = 1000$ (\Cref{fig:regression_n1000}).
The oblivious and semi-oblivious local search algorithms cannot be applied to this setting due to their slow running time.
Moreover, in this setting, we can observe that the non-oblivious variant outperforms the benchmarks.

\subsection{Experiments on Structure Learning of Graphical Models}

\paragraph{Datasets.}
We consider Ising models $G = (V, E)$ with parameter $(w_{uv})_{u, v \in V}$.
First we generate the true graphical model randomly from the configuration model with degree $d$ for all vertices.
For each edge $(u, v) \in E$, we set the parameter $w_{uv} = +0.5$ or $w_{uv} = -0.5$ uniformly at random.
We synthetically generate $100$ samples by Gibbs sampling from this Ising model.
We consider two settings with different parameters.
We set $(|V|, d) = (10, 5)$ in one setting and $(|V|, d) = (20, 7)$ in the other setting.
We consider the problem of maximizing the pseudo-log-likelihood.
For each setting, we conduct $10$ trials and plot the average.

\paragraph{Methods.}
We implement the oblivious, semi-oblivious, and non-oblivious local search algorithms with parameter $q = 1$.
Since calculating $M_{s,3}$ requires much computational cost, we use an upper bound $4 \sum_{i=1}^N \| \bfx^i \|_2^3$ instead of $M_{s,3}$ in the non-oblivious variant.
As a benchmark, we implement modular approximation, in which to maximize a linear function over a $b$-matching constraint, we use the reduction from a max-weight $b$-matching problem to a max-weight matching problem \citep[Theorem 32.4]{Schrijver} and the max-weight matching problem solver in NetwerkX library.
We also implement random selection, which randomly samples a subgraph whose degree is $d$ at all vertices.
In all methods, we use the L-BFGS-G solver in scipy.optimize library for evaluating the value of $f$.
We apply these methods to pseudo-log-likelihood maximization under a $b$-matching constraint for each $b \in \{1, \cdots, d\}$.

\paragraph{Results.}
First, we compare the proposed methods in the case of $|V| = 10$ (\Cref{fig:graphical_n10_time} and \Cref{fig:graphical_n10}).
We can observe that our acceleration techniques work well in practice.
The running time of the non-oblivious variant is competitive with that of modular approximation and its solution quality is higher than that of modular approximation for larger solution size.
Next, we conduct experiments on larger graphs with $|V| = 20$ (\Cref{fig:graphical_n20}).
Since the oblivious and semi-oblivious local search algorithms are too slow to be applied to this setting, we omit them.
Also, in this setting, we can observe that the non-oblivious variant outperforms the benchmarks particularly in cases of larger solution size.

\section*{Acknowledgements}
The author would like to thank Andreas Krause for providing insightful comments in the early stages of this study.
The author is thankful to Takeru Matsuda and Kazuki Matoya for inspiring discussions.
This study was supported by JSPS KAKENHI Grant Number JP 18J12405.

\nocite{langley00}

\bibliography{main}

\begin{thebibliography}{41}
\providecommand{\natexlab}[1]{#1}
\providecommand{\url}[1]{\texttt{#1}}
\expandafter\ifx\csname urlstyle\endcsname\relax
  \providecommand{\doi}[1]{doi: #1}\else
  \providecommand{\doi}{doi: \begingroup \urlstyle{rm}\Url}\fi

\bibitem[Badanidiyuru et~al.(2014)Badanidiyuru, Mirzasoleiman, Karbasi, and
  Krause]{BMKK14}
Badanidiyuru, A., Mirzasoleiman, B., Karbasi, A., and Krause, A.
\newblock Streaming submodular maximization: massive data summarization on the
  fly.
\newblock In \emph{{P}roceedings of the 20th {ACM} {SIGKDD} International
  Conference on Knowledge Discovery and Data Mining ({KDD})}, pp.\  671--680,
  2014.

\bibitem[Bahmani et~al.(2013)Bahmani, Raj, and Boufounos]{BRB13}
Bahmani, S., Raj, B., and Boufounos, P.~T.
\newblock Greedy sparsity-constrained optimization.
\newblock \emph{Journal of Machine Learning Research}, 14\penalty0
  (1):\penalty0 807--841, 2013.

\bibitem[Baldassarre et~al.(2016)Baldassarre, Li, Scarlett, Gozcu, Bogunovic,
  and Cevher]{BLSGB16}
Baldassarre, L., Li, Y., Scarlett, J., Gozcu, B., Bogunovic, I., and Cevher, V.
\newblock Learning-based compressive subsampling.
\newblock \emph{Journal of Selelected Topics in Signal Processing}, 10\penalty0
  (4):\penalty0 809--822, 2016.

\bibitem[Balkanski et~al.(2016)Balkanski, Mirzasoleiman, Krause, and
  Singer]{Balkanski2016}
Balkanski, E., Mirzasoleiman, B., Krause, A., and Singer, Y.
\newblock Learning sparse combinatorial representations via two-stage
  submodular maximization.
\newblock In \emph{Proceedings of The 33rd International Conference on Machine
  Learning (ICML)}, pp.\  2207--2216, 2016.

\bibitem[Baraniuk et~al.(2010)Baraniuk, Cevher, Duarte, and Hegde]{BCDH10}
Baraniuk, R.~G., Cevher, V., Duarte, M.~F., and Hegde, C.
\newblock Model-based compressive sensing.
\newblock \emph{{IEEE} Transactions on Information Theory}, 56\penalty0
  (4):\penalty0 1982--2001, 2010.

\bibitem[Besag(1975)]{Besag75}
Besag, J.
\newblock Statistical analysis of non-lattice data.
\newblock \emph{Journal of the Royal Statistical Society: Series D},
  24\penalty0 (3):\penalty0 179--195, 1975.

\bibitem[Bian et~al.(2017)Bian, Buhmann, Krause, and Tschiatschek]{Bian17}
Bian, A.~A., Buhmann, J.~M., Krause, A., and Tschiatschek, S.
\newblock Guarantees for greedy maximization of non-submodular functions with
  applications.
\newblock In \emph{Proceedings of the 34th International Conference on Machine
  Learning ({ICML})}, pp.\  498--507, 2017.

\bibitem[Bresler(2015)]{Bresler15}
Bresler, G.
\newblock Efficiently learning ising models on arbitrary graphs.
\newblock In \emph{Proceedings of the Forty-Seventh Annual {ACM} on Symposium
  on Theory of Computing ({STOC})}, pp.\  771--782, 2015.

\bibitem[Cevher \& Krause(2011)Cevher and Krause]{CK11}
Cevher, V. and Krause, A.
\newblock Greedy dictionary selection for sparse representation.
\newblock \emph{IEEE Journal of Selected Topics in Signal Processing},
  5\penalty0 (5):\penalty0 979--988, 2011.

\bibitem[Chen et~al.(2018)Chen, Feldman, and Karbasi]{CFK18}
Chen, L., Feldman, M., and Karbasi, A.
\newblock Weakly submodular maximization beyond cardinality constraints: Does
  randomization help greedy?
\newblock In \emph{Proceedings of the 35th International Conference on Machine
  Learning ({ICML})}, pp.\  803--812, 2018.

\bibitem[Das \& Kempe(2011)Das and Kempe]{Das2011}
Das, A. and Kempe, D.
\newblock Submodular meets spectral: Greedy algorithms for subset selection,
  sparse approximation and dictionary selection.
\newblock In \emph{Proceedings of the 28th International Conference on Machine
  Learning (ICML)}, pp.\  1057--1064, 2011.

\bibitem[Elenberg et~al.(2017)Elenberg, Dimakis, Feldman, and
  Karbasi]{Elenberg2017}
Elenberg, E., Dimakis, A.~G., Feldman, M., and Karbasi, A.
\newblock Streaming weak submodularity: Interpreting neural networks on the
  fly.
\newblock In \emph{Advances in Neural Information Processing Systems (NIPS)
  30}, pp.\  4047--4057. 2017.

\bibitem[Elenberg et~al.(2018)Elenberg, Khanna, Dimakis, and
  Negahban]{Elenberg18}
Elenberg, E.~R., Khanna, R., Dimakis, A.~G., and Negahban, S.
\newblock Restricted strong convexity implies weak submodularity.
\newblock \emph{Annals of Statistics}, 46\penalty0 (6B):\penalty0 3539--3568,
  2018.

\bibitem[Feige et~al.(2011)Feige, Mirrokni, and Vondr{\'{a}}k]{FMV11}
Feige, U., Mirrokni, V.~S., and Vondr{\'{a}}k, J.
\newblock Maximizing non-monotone submodular functions.
\newblock \emph{{SIAM} Journal on Computing}, 40\penalty0 (4):\penalty0
  1133--1153, 2011.

\bibitem[Feldman(2013)]{Fel13}
Feldman, M.
\newblock \emph{Maximization Problems with Submodular Objective Functions}.
\newblock PhD thesis, Computer Science Department, Technion - Israel Institute
  of Technology, 2013.

\bibitem[Feldman et~al.(2011)Feldman, Naor, Schwartz, and Ward]{FNSW11}
Feldman, M., Naor, J., Schwartz, R., and Ward, J.
\newblock Improved approximations for k-exchange systems (extended abstract).
\newblock In \emph{Proceedings of the 19th Annual European Symposium on
  Algorithms ({ESA})}, pp.\  784--798, 2011.

\bibitem[Filmus \& Ward(2014)Filmus and Ward]{FW14}
Filmus, Y. and Ward, J.
\newblock Monotone submodular maximization over a matroid via non-oblivious
  local search.
\newblock \emph{{SIAM} Journal on Computing}, 43\penalty0 (2):\penalty0
  514--542, 2014.

\bibitem[Fisher et~al.(1978)Fisher, Nemhauser, and Wolsey]{FNW78}
Fisher, M.~L., Nemhauser, G.~L., and Wolsey, L.~A.
\newblock \emph{An analysis of approximations for maximizing submodular set
  functions---II}, pp.\  73--87.
\newblock Springer Berlin Heidelberg, Berlin, Heidelberg, 1978.

\bibitem[Frank \& Wolfe(1956)Frank and Wolfe]{FW56}
Frank, M. and Wolfe, P.
\newblock An algorithm for quadratic programming.
\newblock \emph{Naval research logistics quarterly}, 3\penalty0 (1-2):\penalty0
  95--110, 1956.

\bibitem[Fujii \& Soma(2018)Fujii and Soma]{FS18}
Fujii, K. and Soma, T.
\newblock Fast greedy algorithms for dictionary selection with generalized
  sparsity constraints.
\newblock In \emph{Advances in Neural Information Processing Systems
  ({NeurIPS}) 31}, pp.\  4749--4758, 2018.

\bibitem[Fujishige(2005)]{Fujishige2005}
Fujishige, S.
\newblock \emph{Submodular Functions and Optimization}.
\newblock Elsevier, 2nd edition, 2005.

\bibitem[Golovin \& Krause(2011)Golovin and Krause]{GK11}
Golovin, D. and Krause, A.
\newblock Adaptive submodularity: Theory and applications in active learning
  and stochastic optimization.
\newblock \emph{Journal of Artificial Intelligence Research}, 42:\penalty0
  427--486, 2011.

\bibitem[Hoi et~al.(2006)Hoi, Jin, Zhu, and Lyu]{HJZL06}
Hoi, S. C.~H., Jin, R., Zhu, J., and Lyu, M.~R.
\newblock Batch mode active learning and its application to medical image
  classification.
\newblock In \emph{Proceedings of the 23rd International Conference of Machine
  Learning {(ICML})}, pp.\  417--424, 2006.

\bibitem[Huang et~al.(2009)Huang, Zhang, and Metaxas]{Huang2009}
Huang, J., Zhang, T., and Metaxas, D.
\newblock Learning with structured sparsity.
\newblock \emph{Journal of Machine Learning Research}, 12:\penalty0 3371--3412,
  2009.

\bibitem[Iyer et~al.(2013)Iyer, Jegelka, and Bilmes]{IJB13}
Iyer, R.~K., Jegelka, S., and Bilmes, J.~A.
\newblock Fast semidifferential-based submodular function optimization.
\newblock In \emph{Proceedings of the 30th International Conference on Machine
  Learning ({ICML})}, pp.\  855--863, 2013.

\bibitem[Jaggi(2013)]{Jaggi13}
Jaggi, M.
\newblock Revisiting frank-wolfe: Projection-free sparse convex optimization.
\newblock In \emph{Proceedings of the 30th International Conference on Machine
  Learning ({ICML})}, pp.\  427--435, 2013.

\bibitem[Jain et~al.(2014)Jain, Tewari, and Kar]{JTK14}
Jain, P., Tewari, A., and Kar, P.
\newblock On iterative hard thresholding methods for high-dimensional
  m-estimation.
\newblock In \emph{Advances in Neural Information Processing Systems ({NIPS})
  27}, pp.\  685--693, 2014.

\bibitem[Jain et~al.(2016)Jain, Rao, and Dhillon]{JRD16}
Jain, P., Rao, N., and Dhillon, I.~S.
\newblock Structured sparse regression via greedy hard thresholding.
\newblock In Lee, D.~D., Sugiyama, M., Luxburg, U.~V., Guyon, I., and Garnett,
  R. (eds.), \emph{Advances in Neural Information Processing Systems 29}, pp.\
  1516--1524, 2016.

\bibitem[Jalali et~al.(2011)Jalali, Johnson, and Ravikumar]{JJR11}
Jalali, A., Johnson, C.~C., and Ravikumar, P.
\newblock On learning discrete graphical models using greedy methods.
\newblock In \emph{Advances in Neural Information Processing Systems ({NIPS})
  24}, pp.\  1935--1943, 2011.

\bibitem[Klivans \& Meka(2017)Klivans and Meka]{KM17}
Klivans, A.~R. and Meka, R.
\newblock Learning graphical models using multiplicative weights.
\newblock In \emph{Proceedings of the 58th {IEEE} Annual Symposium on
  Foundations of Computer Science ({FOCS})}, pp.\  343--354, 2017.

\bibitem[Kyrillidis \& Cevher(2012)Kyrillidis and Cevher]{KC12}
Kyrillidis, A. and Cevher, V.
\newblock Combinatorial selection and least absolute shrinkage via the clash
  algorithm.
\newblock In \emph{Proceedings of the 2012 {IEEE} International Symposium on
  Information Theory ({ISIT})}, pp.\  2216--2220, 2012.

\bibitem[Lacoste{-}Julien \& Jaggi(2015)Lacoste{-}Julien and Jaggi]{LJ15}
Lacoste{-}Julien, S. and Jaggi, M.
\newblock On the global linear convergence of frank-wolfe optimization
  variants.
\newblock In \emph{Advances in Neural Information Processing Systems ({NIPS})
  28}, pp.\  496--504, 2015.

\bibitem[Lee et~al.(2009)Lee, Mirrokni, Nagarajan, and Sviridenko]{LMNS09}
Lee, J., Mirrokni, V.~S., Nagarajan, V., and Sviridenko, M.
\newblock Non-monotone submodular maximization under matroid and knapsack
  constraints.
\newblock In \emph{Proceedings of the 41st Annual {ACM} Symposium on Theory of
  Computing ({STOC})}, pp.\  323--332, 2009.

\bibitem[Lee et~al.(2010)Lee, Sviridenko, and Vondr{\'{a}}k]{LSV10}
Lee, J., Sviridenko, M., and Vondr{\'{a}}k, J.
\newblock Submodular maximization over multiple matroids via generalized
  exchange properties.
\newblock \emph{Mathematics of Operations Research}, 35\penalty0 (4):\penalty0
  795--806, 2010.

\bibitem[Lin \& Bilmes(2011)Lin and Bilmes]{LB11}
Lin, H. and Bilmes, J.~A.
\newblock A class of submodular functions for document summarization.
\newblock In \emph{Proceedings of the 49th Annual Meeting of the Association
  for Computational Linguistics: Human Language Technologies ({ACL})}, pp.\
  510--520, 2011.

\bibitem[Needell \& Tropp(2010)Needell and Tropp]{NT10}
Needell, D. and Tropp, J.~A.
\newblock Cosamp: iterative signal recovery from incomplete and inaccurate
  samples.
\newblock \emph{Communications of the {ACM}}, 53\penalty0 (12):\penalty0
  93--100, 2010.

\bibitem[Negahban et~al.(2012)Negahban, Ravikumar, Wainwright, and Yu]{NRWY12}
Negahban, S.~N., Ravikumar, P., Wainwright, M.~J., and Yu, B.
\newblock A unified framework for high-dimensional analysis of {$M$}-estimators
  with decomposable regularizers.
\newblock \emph{Statistical Science}, 27\penalty0 (4):\penalty0 538--557, 2012.

\bibitem[Nemhauser et~al.(1978)Nemhauser, Wolsey, and Fisher]{NWF78}
Nemhauser, G.~L., Wolsey, L.~A., and Fisher, M.~L.
\newblock An analysis of approximations for maximizing submodular set functions
  - {I}.
\newblock \emph{Mathematical Programming}, 14\penalty0 (1):\penalty0 265--294,
  1978.

\bibitem[Sakaue(2019)]{Sakaue19}
Sakaue, S.
\newblock Greedy and {IHT} algorithms for non-convex optimization with monotone
  costs of non-zeros.
\newblock In \emph{Proceedings of The 22nd International Conference on
  Artificial Intelligence and Statistics ({AISTATS})}, pp.\  206--215, 2019.

\bibitem[Schrijver(2003)]{Schrijver}
Schrijver, A.
\newblock \emph{Combinatorial Optimization: Polyhedra and Efficiency}.
\newblock Springer, Berlin, 2003.

\bibitem[Wu et~al.(2019)Wu, Sanghavi, and Dimakis]{WSD19}
Wu, S., Sanghavi, S., and Dimakis, A.~G.
\newblock Sparse logistic regression learns all discrete pairwise graphical
  models.
\newblock In \emph{Advances in Neural Information Processing Systems
  ({NeurIPS}) 32}, pp.\  8069--8079, 2019.

\end{thebibliography}
\bibliographystyle{icml2020}

\newpage
\onecolumn
\appendix

\section{Omitted pseudocodes}\label{sec:full-pseudocodes}

\begin{algorithm}[h]
	\caption{Local search algorithms for a matroid constraint}\label{alg:matroid-anytime-full}
	\begin{algorithmic}[1]
		\STATE Let $X \gets \emptyset$.
		\STATE Add arbitrary elements to $X$ until $X$ is maximal in $\calI$.
		\FOR{$i = 1,\cdots,T$}
			\STATE Determine the pair of $x \in X$ and $x' \in N \setminus X$ by the following rules:\\
			$\begin{cases}
				\displaystyle (x, x') \in \argmax \{ f(X - x + x') \mid X - x + x' \in \calI \} & (\text{oblivious})\\
				\displaystyle \text{Let $\displaystyle x' \in \argmax \{ f(X - \phi_X(x') + x') - f(X) \}$ and $x = \phi_X(x')$,}\\
				\quad \displaystyle \text{where $\phi_X \colon N \setminus X \to X$ is a map that satisfies}\\
				\quad \phi_X(x') \in \argmin_{x \in X \colon X - x + x' \in \calI} (\bfw^{(X)})_x^2 & (\text{semi-oblivious})\\
				(x, x') \in \argmax_{(x, x') \colon X - x + x'} \left\{ \frac{1}{2M_{s,2}} \left( \nabla u(\bfw^{(X)}) \right)^2_{x'} - \frac{M_{s,2}}{2} \left(\bfw^{(X)}\right)_{x}^2  \right\} & (\text{non-oblivious})
			\end{cases}
				$
			\IF{$
			\begin{cases}
				f(X - x + x') - f(X) > 0 & (\text{oblivious or semi-oblivious})\\
				\frac{1}{2M_{s,2}} \left( \nabla u(\bfw^{(X)}) \right)^2_{x'} - \frac{M_{s,2}}{2} \left(\bfw^{(X)}\right)_{x}^2 > 0 & (\text{non-oblivious})
			\end{cases}
				$}
				\STATE Update the solution $X \gets X - x + x'$.
			\ELSE
				\STATE \textbf{return} $X$.
			\ENDIF
		\ENDFOR
		\STATE \textbf{return} $X$.
	\end{algorithmic}
\end{algorithm}

\begin{algorithm}[H]
	\caption{Local search algorithms for a $p$-matroid intersection or $p$-exchange system ($p \ge 2$)}\label{alg:system-anytime-full}
	\begin{algorithmic}[1]
		\STATE Let $t = \begin{cases}
			2p(q+1) & \text{in the case of $p$-matroid intersection constraints}\\
			pq + 1 & \text{in the case of $p$-exchange system constraints}.
						\end{cases}$
		\STATE Let $X \gets \emptyset$.
		\STATE Add arbitrary elements to $X$ until $X$ is maximal in $\calI$.
		\FOR{$i = 1,\cdots,T$}
			\STATE Determine $X'$ that is $q$-reachable from $X$ such that:\\
			$\begin{cases}
				\displaystyle X' \in \argmax_{X' \in \calF_q(X)} f(X') & (\text{oblivious})\\
				\displaystyle \text{Let $X' \in \argmax_{X' \in \calF_q(X) \colon \exists S, ~ X' = (X \cup S) \setminus \phi_X(S)} f(X')$,}\\
				\quad \displaystyle \text{where $\phi_X \colon 2^N \to 2^N$ is a map that satisfies}\\
				\quad \phi_X(S) \in \argmin_{T \colon (X \cup S) \setminus T \in \calF_q(X)} \| (\bfw^{(X)})_T \|^2 & (\text{semi-oblivious})\\
				X' \in \argmax_{X' \in \calF_q(X)} \left\{ \frac{1}{2M_{s,t}} \left\| \left( \nabla u(\bfw^{(X)}) \right)_{X' \setminus X} \right\|^2 - \frac{M_{s,t}}{2} \left\| \left(\bfw^{(X)}\right)_{X \setminus X'} \right\|^2  \right\} & (\text{non-oblivious})
			\end{cases}
				$
			\IF{$
			\begin{cases}
				f(X') - f(X) > 0 & (\text{oblivious or semi-oblivious})\\
				\frac{1}{2M_{s,t}} \left\| \left( \nabla u(\bfw^{(X)}) \right)_{X' \setminus X} \right\|^2 - \frac{M_{s,t}}{2} \left\| \left(\bfw^{(X)}\right)_{X \setminus X'} \right\|^2 > 0 & (\text{non-oblivious})
			\end{cases}
				$}
				\STATE Update the solution $X \gets X'$.
			\ELSE
				\STATE \textbf{return} $X$.
			\ENDIF
		\ENDFOR
		\STATE \textbf{return} $X$.
	\end{algorithmic}
\end{algorithm}

\section{Omitted Proofs}

\subsection{Properties of Matroids, $p$-Matroid Intersection, and $p$-Exchange Systems}
Here we provide important lemmas used in the proofs.

The following lemma is the exchange property of matroids.
\begin{lemma}[{Corollary 39.12a in~\citet{Schrijver}}]\label{lem:schrijver}
	Let $\calM = (N, \calI)$ be a matroid and $I, J \in \calI$ with $|I| = |J|$.
	There exists a bijection $\varphi \colon I \setminus J \to J \setminus I$ such that $I - v + \varphi(v) \in \calI$ for all $v \in I \setminus J$.
\end{lemma}

The following lemma is on the exchange property of $p$-matroid intersection, which was first used for analyzing local search algorithms for submodular maximization.
\begin{lemma}[{\citep{LSV10}}]\label{lem:multiset-intersection}
	Suppose $\calI$ is a $p$-matroid intersection.
	Let $q \in \bbZ$ be any positive integer.
	For any $S, T \in \calI$, there exists a multiset $\calP \subseteq 2^N$ and an integer $\eta$ (depending on $p$ and $q$) that satisfies the following conditions.
	\begin{enumerate}
		\item For all $P \in \calP$, the symmetric difference is feasible, i.e., $S \triangle P \in \calI$, and $S \triangle P$ is $q$-reachable (\Cref{def:reachability-intersection}) from $S$.
		\item Each element $v \in T \setminus S$ appears in exactly $q \eta$ sets in $\calP$.
		\item Each element $v \in S \setminus T$ appears in at most $(pq - q + 1) \eta$ sets in $\calP$.
	\end{enumerate}
\end{lemma}

A property similar to that for $p$-matroid intersection was known for $p$-exchange systems as follows.
\begin{lemma}[{\citep{FNSW11}; The full proof can be found in \citet{Fel13}}]\label{lem:multiset-exchange}
	Suppose $\calI$ is a $p$-exchange system.
	Let $q \in \bbZ$ be any positive integer.
	For any $S, T \in \calI$, there exists a multiset $\calP \subseteq 2^N$ and an integer $\eta$ (depending on $p$ and $q$) that satisfies the following conditions.
	\begin{enumerate}
		\item For all $P \in \calP$, the symmetric difference is feasible, i.e., $S \triangle P \in \calI$, and $S \triangle P$ is $q$-reachable (\Cref{def:reachability-system}) from $S$.
		\item Each element $v \in T \setminus S$ appears in at most $q\eta$ sets in $\calP$.
		\item Each element $v \in S \setminus T$ appears in at most $(pq-q+1)\eta$ sets in $\calP$.
	\end{enumerate}
\end{lemma}

\subsection{Proof of the Localizability of Sparse Optimization}
To prove the localizability of sparse optimization, we use the following lemmas.

\begin{lemma}\label{lem:feature-smooth}
	Suppose $u \colon 2^N \to \bbR$ is a continuously differentiable function with $u(\bfzero) \ge 0$.
	Assume $u$ is restricted strong concave on $\Omega_{2s}$ and restricted smooth on $\Omega_{s,t}$.
	If $f \colon 2^N \to \bbR$ is a set function defined as $f(X) = \max_{\supp(\bfw) \subseteq X} u(\bfw)$, then for any $X, X' \subseteq N$ with $s = \max \{ |X|, |X'| \}$ and $t = |X \triangle X'|$, we have
	\begin{equation*}
		f(X') - f(X) \ge \frac{1}{2 M_{s, t}} \left\| \left(\nabla u(\bfw^{(X)})\right)_{X' \setminus X} \right\|^2 - \frac{M_{s, t}}{2} \left\| \left(\bfw^{(X)}\right)_{X \setminus X'} \right\|^2.
	\end{equation*}
\end{lemma}

\begin{proof}
	From the restricted smoothness of $u$, for any $\bfz \in \bbR^n$ with $\supp(\bfz) \subseteq X' \setminus X$, we have
    \begin{align*}
        f(X') - f(X)
        &= u (\bfw^{(X')}) - u (\bfw^{(X)})\\
        &\ge u ((\bfw^{(X)})_{X \cap X'} + \bfz) - u (\bfw^{(X)})\\
        &\ge \left\langle \nabla u (\bfw^{(X)}), \bfz - (\bfw^{(X)})_{X \setminus X'} \right\rangle - \frac{M_{s,t}}{2} \left\| \bfz - (\bfw^{(X)})_{X \setminus X'} \right\|^2.
    \end{align*}
	Since this inequality holds for every $\bfz$ with $\supp(\bfz) \subseteq X' \setminus X$, by optimizing it for $\bfz$, we obtain
    \begin{equation*}
        f(X') - f(X)
        \ge \frac{1}{2M_{s,t}} \left\| \nabla u(\bfw^{(X)})_{X' \setminus X} \right\|^2 - \frac{M_{s,t}}{2} \left\| (\bfw^{(X)})_{X \setminus X'} \right\|^2.
    \end{equation*}
\end{proof}

\begin{lemma}\label{lem:feature-concave}
	Suppose $u \colon 2^N \to \bbR$ is a continuously differentiable function with $u(\bfzero) \ge 0$.
	Assume $u$ is restricted strong concave on $\Omega_{2s}$ and restricted smooth on $\Omega_{s,t}$.
	If $f \colon 2^N \to \bbR$ is a set function defined as $f(X) = \max_{\supp(\bfw) \subseteq X} u(\bfw)$, then for any $X, X' \subseteq N$ with $s = \max \{ |X|, |X^*| \}$, we have
	\begin{equation*}
		f(X^*) - f(X) \le \frac{1}{2 m_{2s}} \left\| \left(\nabla u(\bfw^{(X)})\right)_{X^* \setminus X} \right\|^2 - \frac{m_{2s}}{2} \left\| \left(\bfw^{(X)}\right)_{X \setminus X^*} \right\|^2.
	\end{equation*}
\end{lemma}

\begin{proof}
	From the restricted strong concavity of $u$, we obtain
    \begin{align*}
        f(X^*) - f(X)
        &= u(\bfw^{(X^*)}) - u(\bfw^{(X)}) \nonumber\\
        &\le \left\langle \nabla u (\bfw^{(X)}), \bfw^{(X^*)} - \bfw^{(X)} \right\rangle - \frac{m_{2s}}{2} \left\| \bfw^{(X^*)} - \bfw^{(X)} \right\|^2 \nonumber \\
        &\le \max_{\bfz \colon \supp(\bfz) \subseteq X^*} \left\{ \left\langle \nabla u (\bfw^{(X)}), \bfz - \bfw^{(X)} \right\rangle - \frac{m_{2s}}{2} \left\| \bfz - \bfw^{(X)} \right\|^2 \right\} \nonumber \\
        &= \frac{1}{2m_{2s}} \left\| \left( \nabla u (\bfw^{(X)}) \right)_{X^* \setminus X} \right\|^2 - \frac{m_{2s}}{2} \left\| \left( \bfw^{(X)} \right)_{X \setminus X^*} \right\|^2.
    \end{align*}
\end{proof}

Now we prove \Cref{lem:feature-exchange} from the above two lemmas.

\begin{proof}[Proof of \Cref{lem:feature-exchange}]
	From \Cref{lem:feature-smooth}, we have
	\begin{equation*}
		f(X \triangle P) - f(X)
		\ge \frac{1}{2M_{s,t}} \left\| \left( \nabla u(\bfw^{(X)}) \right)_{P \setminus X} \right\|^2 - \frac{M_{s,t}}{2} \left\| \left( \bfw^{(X)} \right)_{P \cap X} \right\|^2
	\end{equation*}
	for all $P \in \calP$.
	By adding this inequality for each $P \in \calP$, we obtain 
	\begin{equation*}
		\sum_{P \in \calP} \left\{ f(X \triangle P) - f(X) \right\}
		\ge k \frac{1}{2M_{s,t}} \left\| \left( \nabla u(\bfw^{(X)}) \right)_{X^* \setminus X} \right\|^2 - \ell \frac{M_{s,t}}{2} \left\| \left( \bfw^{(X)} \right)_{X \setminus X^*} \right\|^2,
	\end{equation*}
	where we used the fact that each element in $v \in X^* \setminus X$ appears in at least $k$ sets in $\calP$ and each element in $X \setminus X^*$ appears in at most $\ell$ sets in $\calP$.
	From the strong concavity of $u$, by applying \Cref{lem:feature-concave}, we obtain
	\begin{equation*}
		f(X^*) - f(X) \le \frac{1}{2m_{2s}} \left\| \left( \nabla u (\bfw^{(X)}) \right)_{X^* \setminus X} \right\|^2 - \frac{m_{2s}}{2} \left\| \left( \bfw^{(X)} \right)_{X \setminus X^*} \right\|^2
	\end{equation*}
	and 
	\begin{align*}
		- f(X) 
		&\le f(\emptyset) - f(X) \\
		&\le - \frac{m_{2s}}{2} \left\| \left( \bfw^{(X)} \right)_{X} \right\|^2\\
		&\le - \frac{m_{2s}}{2} \left\| \left( \bfw^{(X)} \right)_{X \setminus X^*} \right\|^2.
	\end{align*}
	By adding the former inequality multiplied by $km_{2s}/M_{s,t}$ and the latter inequality multiplied by $\ell M_{s,t}/m_{2s}-km_{2s}/M_{s,t}$, we obtain
	\begin{equation*}
		\sum_{P \in \calP} \left\{ f(X \triangle P) - f(X) \right\}
		\ge \frac{m_{2s}}{M_{s,t}} k f(X^*) - \frac{M_{s,t}}{m_{2s}} \ell f(X).
	\end{equation*}
\end{proof}

\subsection{Proof of \Cref{thm:matroid-anytime}}

\begin{proof}[Proof of \Cref{thm:matroid-anytime}]
	Let $X$ be the output of the algorithm and $X^*$ an optimal solution.
	From \Cref{lem:schrijver}, we have a bijection $\phi \colon X^* \setminus X \to X \setminus X^*$ such that $X - \phi(x^*) + x^* \in \calI$ for all $x^* \in X^* \setminus X$.

	Suppose at some iteration the solution is updated from $X$ to $X - x + x'$.
	Since $\displaystyle (x, x') \in \argmax \{ f(X - x + x') \mid X - x + x' \in \calI \}$ and $f(X - x + x') - f(X) > 0$, we have
	\begin{align*}
		f(X - x + x') - f(X)
		&= \max_{(x, x') \colon X - x + x' \in \calI} f(X - x + x') - f(X) \\
		&\ge \frac{1}{s} \sum_{x^* \in X^* \setminus X} \left\{ f(X - \phi(x^*) + x^*) - f(X) \right\}.
	\end{align*}
	By setting $\calP = \{ \{x, \phi(x) \} \mid x \in X \setminus X^* \}$, each element in $X^* \setminus X$ and $X \setminus X^*$ appears exactly once in $\calP$.
	Since $f$ is $(\alpha, \beta_1, \beta_2)$-localizable with size $s$ and exchange size $2$, we obtain
	\begin{equation*}
		\sum_{x^* \in X^* \setminus X} \left\{ f(X - \phi(x^*) + x^*) - f(X) \right\}
		\ge \alpha f(X^*) - (\beta_1 + \beta_2) f(X).
	\end{equation*}
	By combining these inequalities, we have
	\begin{align*}
		f(X - x + x') - f(X)
		&\ge \frac{1}{s} \left\{ \alpha f(X^*) - (\beta_1 + \beta_2) f(X) \right\}\\
		&= \frac{\beta_1 + \beta_2}{s} \left\{ \frac{\alpha}{\beta_1 + \beta_2} f(X^*) - f(X) \right\},
	\end{align*}
	which implies that the distance from the current solution to $\alpha / (\beta_1 + \beta_2)$ times the optimal value is decreased by the rate $1 - (\beta_1 + \beta_2) / s$ at each iteration.
	Hence, the approximation ratio after $T$ iterations can be bounded as
	\begin{align*}
		f(X)
		&\ge \frac{\alpha}{\beta_1 + \beta_2} \left( 1 - \left( 1 - \frac{\beta_1 + \beta_2}{s} \right)^T \right) f(X^*)\\
		&\ge \frac{\alpha}{\beta_1 + \beta_2} \left( 1 - \exp\left( - \frac{(\beta_1 + \beta_2) T}{s} \right) \right) f(X^*),
	\end{align*}
	which proves the first statement of the theorem.

	Next, we consider the case where the algorithm stops by finding no pair to improve the objective value.
	When the oblivious variant stops, we have $f(X) \ge f(X - x + x')$ for all $x \in X$ and $x' \in N \setminus X$ such that $X - x + x' \in \calI$.
	In the same manner as the above analysis, we obtain
	\begin{align*}
		0 &\ge \sum_{x^* \in X^* \setminus X} \left\{ f(X - \phi(x^*) + x^*) - f(X) \right\} \\
		  &\ge \alpha f(X^*) - (\beta_1 + \beta_2) f(X),
	\end{align*}
	which implies
	\begin{equation*}
		f(X)
		\ge \frac{\alpha}{\beta_1 + \beta_2} f(X^*).
	\end{equation*}
\end{proof}

\subsection{Proof of \Cref{thm:system-anytime}}

\begin{proof}[Proof of \Cref{thm:system-anytime}]
	Let $X$ be the output of the algorithm and $X^*$ an optimal solution.
	From \Cref{lem:multiset-intersection} and \Cref{lem:multiset-exchange} for each case, respectively, we can observe that there exists a multiset $\calP \subseteq 2^N$ and an integer $\eta$ that satisfy the following conditions.
	\begin{enumerate}
		\item For all $P \in \calP$, the symmetric difference is $q$-reachable from $X$, i.e., $X \triangle P \in \calF_q(X)$.
		\item Each element $v \in X^* \setminus X$ appears in exactly $q \eta$ sets in $\calP$.
		\item Each element $v \in X \setminus X^*$ appears in at most $(pq - q + 1) \eta$ sets in $\calP$.
	\end{enumerate}

	Suppose at some iteration the solution is updated from $X$ to $X'$.
	Since $\displaystyle X' \in \argmax_{X' \in \calF_q(X)} f(X')$ and $f(X') - f(X) > 0$, we have
	\begin{align*}
		f(X') - f(X)
		&= \max_{X' \in \calF_q(X)} f(X') - f(X) \\
		&\ge \frac{1}{|\calP|} \sum_{P \in \calP} \left\{ f(X \triangle P) - f(X) \right\}.
	\end{align*}
	Since $f$ is $(\alpha, \beta_1, \beta_2)$-localizable with size $s$ and exchange size $t$, we have
	\begin{equation*}
		\sum_{P \in \calP} \left\{ f(X \triangle P) - f(X) \right\}
		\ge \alpha q \eta f(X^*) - \left( \beta_1 (pq - q + 1) \eta + \beta_2 q \eta \right) f(X).
	\end{equation*}
	By combining these inequalities, we have
	\begin{equation*}
		f(X') - f(X)
		\ge \frac{1}{|\calP|} \left\{ \alpha q \eta f(X^*) - \left( \beta_1 (pq -q + 1) + \beta_2 q\right) \eta f(X) \right\}.
	\end{equation*}
	Since each element in $T \setminus S$ appears in $q\eta$ sets in $\calP$ and $|T \setminus S| \le s$, it holds that $|\calP| \le sq \eta$.
	Hence, we obtain
	\begin{align*}
		f(X') - f(X)
		&\ge \frac{1}{s} \left\{ \alpha f(X^*) - \left( \beta_1 (p - 1 + 1/q) + \beta_2 \right) f(X) \right\}\\
		&= \frac{\beta_1 (p - 1 + 1/q) + \beta_2}{s} \left\{ \frac{\alpha}{\beta_1 (p - 1 + 1/q) + \beta_2 } f(X^*) - f(X) \right\},
	\end{align*}
	which implies that the distance from the current solution to $\frac{\alpha}{\beta_1 (p - 1 + 1/q) + \beta_2}$ times the optimal value decreases by the rate $1 - \frac{\beta_1 (p - 1 + 1/q) + \beta_2}{s}$ at each iteration.
	Therefore, the solution $X$ after $T$ iterations satisfies
	\begin{equation*}
		f(X) \ge \frac{\alpha}{\beta_1 (p - 1 + 1/q) + \beta_2 } \left( 1 - \exp\left( \frac{(\beta_1 (p - 1 + 1/q) + \beta_2 )  T}{s}\right) \right) f(X^*).
	\end{equation*}

	Next, we consider the case where the algorithm stops by finding no pair to improve the objective value.
	In this case, we have $f(X) \ge f(X')$ for all $X' \in \calF_q(X)$.
	In the same manner as the above analysis, we obtain
	\begin{align*}
		0 &\ge \sum_{P \in \calP} \left\{ f(X \triangle P) - f(X) \right\} \\
		&\ge \frac{1}{s} \left\{ \alpha f(X^*) - \left( \beta_1 (p - 1 + 1/q) + \beta_2 \right) f(X) \right\},
	\end{align*}
	which implies
	\begin{equation*}
		f(X) \ge \frac{\alpha}{\beta_1 (p - 1 + 1/q) + \beta_2 } f(X^*).
	\end{equation*}
\end{proof}

\subsection{Proof of \Cref{thm:matroid-sparse}}

\begin{proof}[Proof of \Cref{thm:matroid-sparse}]
	Let $X$ be the output of the algorithm and $X^*$ an optimal solution.
	Suppose at some iteration the solution is updated from $X$ to $X - x + x'$.
	From \Cref{lem:schrijver}, we have a bijection $\phi \colon X^* \setminus X \to X \setminus X^*$ such that $X - \phi(x^*) + x^* \in \calI$ for all $x^* \in X^* \setminus X$.
	Here we show that
	\begin{equation*}
		f(X - x + x') - f(X) \ge \frac{1}{n} \frac{M_{s,2}}{m_{2s}} \left\{ \frac{m_{2s}^2}{M_{s,2}^2} f(X^*) - f(X) \right\}
	\end{equation*}
	holds at each iteration of the semi-oblivious and non-oblivious variants.

	When using the semi-oblivious variant, due to the property of the algorithm, we have
	\begin{equation*}
		(\bfw^{(X)})^2_{\tilde{x}} \ge (\bfw^{(X)})^2_x
	\end{equation*}
	for any $\tilde{x} \in X$ such that $X - \tilde{x} + x' \in \calI$.
	If $\phi_X \colon N \setminus X \to X$ is a map defined as $\phi_X(x') \in \argmin_{x \in X} \{ (\bfw^{(X)})_x^2 \mid X - x + x' \in \calI \}$, then
	\begin{align*}
		&f(X - x + x') - f(X)\\
		&= \max_{x' \in N \setminus X} f(X - \phi_X(x') + x') - f(X)\\
		&\ge \frac{1}{s} \sum_{x^* \in X^* \setminus X} \{ f(X - \phi_X(x^*) + x^*) - f(X) \} \\
		&\ge \frac{1}{s} \sum_{x^* \in X^* \setminus X} \left\{ \frac{1}{2M_{s,2}} \left( \nabla u(\bfw^{(X)}) \right)^2_{x^*} - \frac{M_{s,2}}{2} \left(\bfw^{(X)}\right)_{\phi_X(x^*)}^2 \right\} \tag{From \Cref{lem:feature-smooth}} \\
		&\ge \frac{1}{s} \sum_{x^* \in X^* \setminus X} \left\{ \frac{1}{2M_{s,2}} \left( \nabla u(\bfw^{(X)}) \right)^2_{x^*} - \frac{M_{s,2}}{2} \left(\bfw^{(X)}\right)_{\phi(x^*)}^2 \right\} \tag{since $\left(\bfw^{(X)}\right)_{\phi(x^*)}^2 \ge \left(\bfw^{(X)}\right)_{\phi_X(x^*)}^2$}\\
		&= \frac{1}{s} \left\{ \frac{1}{2M_{s,2}} \left\| \left( \nabla u(\bfw^{(X)}) \right)_{X^* \setminus X} \right\|^2 - \frac{M_{s,2}}{2} \left\| \left(\bfw^{(X)}\right)_{X \setminus X^*} \right\|^2 \right\} \\
		&\ge \frac{1}{s} \left\{ \frac{m_{2s}}{M_{s,2}} f(X^*) - \frac{M_{s,2}}{m_{2s}} f(X) \right\} \tag{From \Cref{lem:feature-concave}},
	\end{align*}
	where we used \Cref{lem:feature-smooth} and \Cref{lem:feature-concave} as in the oblivious case.

	When using the non-oblivious variant, we have
	\begin{align*}
		&f(X - x + x') - f(X)\\
		&\ge \frac{1}{2M_{s,2}} \left( \nabla u(\bfw^{(X)}) \right)^2_{x'} - \frac{M_{s,2}}{2} \left(\bfw^{(X)}\right)_{x}^2 \tag{From \Cref{lem:feature-smooth}} \\
		&= \max_{(x,x') \colon X - x + x' \in \calI} \left\{ \frac{1}{2M_{s,2}} \left( \nabla u(\bfw^{(X)}) \right)^2_{x'} - \frac{M_{s,2}}{2} \left(\bfw^{(X)}\right)_{x}^2 \right\}\\
		&\ge \frac{1}{s} \sum_{x^* \in X^* \setminus X} \left\{ \frac{1}{2M_{s,2}} \left( \nabla u(\bfw^{(X)}) \right)^2_{x^*} - \frac{M_{s,2}}{2} \left(\bfw^{(X)}\right)_{\phi(x^*)}^2 \right\} \\
		&= \frac{1}{s} \left\{ \frac{1}{2M_{s,2}} \left\| \left( \nabla u(\bfw^{(X)}) \right)_{X^* \setminus X} \right\|^2 - \frac{M_{s,2}}{2} \left\| \left(\bfw^{(X)}\right)_{X \setminus X^*} \right\|^2 \right\} \\
		&\ge \frac{1}{s} \left\{ \frac{m_{2s}}{M_{s,2}} f(X^*) - \frac{M_{s,2}}{m_{2s}} f(X) \right\} \tag{From \Cref{lem:feature-concave}}.
	\end{align*}

	Therefore, in semi-oblivious and non-oblivious variants, we have 
	\begin{align*}
		f(X - x + x') - f(X)
		&\ge \frac{1}{s} \left\{ \frac{m_{2s}}{M_{s,2}} f(X^*) - \frac{M_{s,2}}{m_{2s}} f(X) \right\}\\
		&= \frac{1}{s} \frac{M_{s,2}}{m_{2s}} \left\{ \frac{m_{2s}^2}{M_{s,2}^2} f(X^*) - f(X) \right\},
	\end{align*}
	which implies that the distance from the current solution to $m_{2s}^2 / M_{s,2}^2$ times the optimal value decreases by the rate $1 - M_{s,2} / (s m_{2s})$ at each iteration.
	Hence, the approximation ratio after $T$ iterations can be bounded as
	\begin{align*}
		f(X)
		&\ge \frac{m^2_{2s}}{M^2_{s,2}} \left( 1 - \left( 1 - \frac{M_{s,2}}{s m_{2s}} \right)^T \right) f(X^*)\\
		&\ge \frac{m^2_{2s}}{M^2_{s,2}} \left( 1 - \exp\left( - \frac{M_{s,2} T}{s m_{2s}} \right) \right) f(X^*),
	\end{align*}
	which proves the first statement of the theorem.

	Next, we consider the case where the algorithm stops by finding no pair to improve the objective value.
	For the semi-oblivious and non-oblivious variants, we show that
	\begin{equation*}
		0 \ge \frac{m_{2s}}{M_{s,2}} f(X^*) - \frac{M_{s,2}}{m_{2s}} f(X)
	\end{equation*}
	holds when the algorithm stops, from which the second statement of the theorem follows.
	When the semi-oblivious variant stops, we have $f(X) \ge f(X - \phi_X(x') + x')$ for all $x' \in N \setminus X$, where $\phi_X(x')$ is defined in the algorithm.
	Hence, in the same manner as the above analysis, we obtain
	\begin{align*}
		0 &\ge \sum_{x^* \in X^* \setminus X} \left\{ f(X - \phi_X(x^*) + x^*) - f(X) \right\} \\
		&\ge \frac{m_{2s}}{M_{s,2}} f(X^*) - \frac{M_{s,2}}{m_{2s}} f(X).
	\end{align*}
	When the non-oblivious variant stops, we have 
	\begin{equation*}
		0 \ge \frac{1}{2M_{s,2}} \left( \nabla u(\bfw^{(X)}) \right)^2_{x'} - \frac{M_{s,2}}{2} \left(\bfw^{(X)}\right)_{x}^2
	\end{equation*}
	for all $x \in X$ and $x' \in N \setminus X$ such that $X - x + x' \in \calI$.
	Therefore, we have
	\begin{align*}
		0 &\ge \sum_{x^* \in X^* \setminus X} \left\{ \frac{1}{2M_{s,2}} \left( \nabla u(\bfw^{(X)}) \right)^2_{x'} - \frac{M_{s,2}}{2} \left(\bfw^{(X)}\right)_{x}^2 \right\}\\
		&\ge \frac{m_{2s}}{M_{s,2}} f(X^*) - \frac{M_{s,2}}{m_{2s}} f(X)
	\end{align*}
	by using the above analysis for the first statement.
\end{proof}

\subsection{Proof of \Cref{thm:system-sparse}}

\begin{proof}[Proof of \Cref{thm:system-sparse}]
	Let $X$ be the output of the algorithm and $X^*$ an optimal solution.
	Suppose at some iteration the solution is updated from $X$ to $X'$.
	From \Cref{lem:multiset-intersection} and \Cref{lem:multiset-exchange} for each case, respectively, we can observe that there exist a multiset $\calP \subseteq 2^N$ and an integer $\eta$ that satisfy the following conditions.
	\begin{enumerate}
		\item For all $P \in \calP$, the symmetric difference is $q$-reachable from $X$, i.e., $X \triangle P \in \calF_q(X)$.
		\item Each element $v \in X^* \setminus X$ appears in exactly $q \eta$ sets in $\calP$.
		\item Each element $v \in X \setminus X^*$ appears in at most $(pq - q + 1) \eta$ sets in $\calP$.
	\end{enumerate}
	Here we show that
	\begin{equation*}
		f(X') - f(X)
		\ge \frac{1}{s} \left\{ \frac{m_{2s}}{M_{s,t}} f(X^*) - (p - 1 + 1/q) \frac{M_{s,t}}{m_{2s}} f(X) \right\}.
	\end{equation*}
	holds at each iteration of the semi-oblivious and non-oblivious variants.

	When using the semi-oblivious variant, due to the property of the algorithm, we have
	\begin{equation*}
		\| (\bfw^{(X)})_{T} \|^2 \ge \| (\bfw^{(X)})_{X \setminus X'} \|^2
	\end{equation*}
	for any $T \subseteq X$ such that $(X \cup X') \setminus T \in \calF_q(X)$.
	If $\phi_X \colon 2^N \to 2^N$ is a map defined as $\phi_X(S) \in \argmin_{T \colon (X \cup S) \setminus T \in \calF_q(X)} \| (\bfw^{(X)})_T \|^2$, then
	\begin{align*}
		&f(X') - f(X)\\
		&= \max_{X' \in \calF_q(X) \colon \exists S, ~ X' = (X \cup S) \setminus \phi_X (S)} f(X') - f(X)\\
		&\ge \frac{1}{|\calP|} \sum_{P \in \calP} \{ f((X \cup P) \setminus \phi_X(P \setminus X)) - f(X) \} \\
		&\ge \frac{1}{|\calP|} \sum_{P \in \calP} \left\{ \frac{1}{2M_{s,t}} \left\| \left( \nabla u(\bfw^{(X)}) \right)_{P \setminus X} \right\|^2 - \frac{M_{s,t}}{2} \left\| \left( \bfw^{(X)} \right)_{\phi_X(P \setminus X)} \right\|^2 \right\} \tag{From \Cref{lem:feature-smooth}} \\
		&\ge \frac{1}{|\calP|} \sum_{P \in \calP} \left\{ \frac{1}{2M_{s,t}} \left\| \left( \nabla u(\bfw^{(X)}) \right)_{P \setminus X} \right\|^2 - \frac{M_{s,t}}{2} \left\| \left( \bfw^{(X)} \right)_{P \cap X} \right\|^2 \right\} \tag{since $\left\| \left(\bfw^{(X)}\right)_{P \cap X} \right\|^2 \ge \left\| \left(\bfw^{(X)}\right)_{\phi_X(P \setminus X)} \right\|^2$}\\
		&\ge \frac{1}{|\calP|} \left\{ q \eta \frac{1}{2M_{s,t}} \left\| \left( \nabla u(\bfw^{(X)}) \right)_{X^* \setminus X} \right\|^2 - (pq - q + 1) \eta \frac{M_{s,t}}{2} \left\| \left( \bfw^{(X)} \right)_{X \setminus X^*} \right\|^2 \right\} \\
		&\ge \frac{1}{|\calP|} \left\{ q \eta \frac{m_{2s}}{M_{s,t}} f(X^*) - (pq - q + 1) \eta \frac{M_{s,t}}{m_{2s}} f(X) \right\} \tag{From \Cref{lem:feature-concave}} \\
		&\ge \frac{1}{s} \left\{ \frac{m_{2s}}{M_{s,t}} f(X^*) - (p - 1 + 1/q) \frac{M_{s,t}}{m_{2s}} f(X) \right\},
	\end{align*}
	where we used \Cref{lem:feature-smooth} and \Cref{lem:feature-concave} as in the oblivious case.

	When using the non-oblivious variant, we have
	\begin{align*}
		&f(X') - f(X)\\
		&\ge \frac{1}{2M_{s,t}} \left\| \left( \nabla u(\bfw^{(X)}) \right)_{X' \setminus X} \right\|^2 - \frac{M_{s,t}}{2} \left\| \left( \bfw^{(X)} \right)_{X \setminus X'} \right\|^2\\
		&= \max_{X' \in \calF_q(X)} \left\{ \frac{1}{2M_{s,t}} \left\| \left( \nabla u(\bfw^{(X)}) \right)_{X' \setminus X} \right\|^2 - \frac{M_{s,t}}{2} \left\| \left( \bfw^{(X)} \right)_{X \setminus X'} \right\|^2 \right\} \\
		&\ge \frac{1}{|\calP|} \sum_{P \in \calP} \left\{ \frac{1}{2M_{s,t}} \left\| \left( \nabla u(\bfw^{(X)}) \right)_{P \setminus X} \right\|^2 - \frac{M_{s,t}}{2} \left\| \left( \bfw^{(X)} \right)_{P \cap X} \right\|^2 \right\} \\
		&\ge \frac{1}{|\calP|} \left\{ q \eta \frac{1}{2M_{s,t}} \left\| \left( \nabla u(\bfw^{(X)}) \right)_{X^* \setminus X} \right\|^2 - (pq - q + 1) \eta \frac{M_{s,t}}{2} \left\| \left( \bfw^{(X)} \right)_{X \setminus X^*} \right\|^2 \right\} \\
		&\ge \frac{1}{|\calP|} \left\{ q \eta \frac{m_{2s}}{M_{s,t}} f(X^*) - (pq - q + 1) \eta \frac{M_{s,t}}{m_{2s}} f(X) \right\} \\
		&\ge \frac{1}{s} \left\{ \frac{m_{2s}}{M_{s,t}} f(X^*) - (p - 1 + 1/q) \frac{M_{s,t}}{m_{2s}} f(X) \right\},
	\end{align*}
	where we used $|\calP| \le sq \eta$ in the last inequality.
	Therefore, in semi-oblivious and non-oblivious variants, we have 
	\begin{equation*}
		f(X') - f(X)
		\ge (p - 1 + 1/q) \frac{M_{s,t}}{s m_{2s}} \left\{ \frac{1}{p - 1 + 1/q} \frac{m_{2s}^2}{M_{s,t}^2} f(X^*) -  f(X) \right\}.
	\end{equation*}
	which implies that the distance from the current solution to $\frac{1}{p - 1 + 1/q} \frac{m_{2s}^2}{M_{s,2}^2}$ times the optimal value decreases by the rate $1 - (p-1+1/q)m_{2s} / (sM_{s,2})$ at each iteration.
	Hence, the approximation ratio after $T$ iterations can be bounded as
	\begin{align*}
		f(X)
		&\ge \frac{1}{p - 1 + 1/q} \frac{m^2_{2s}}{M^2_{s,2}} \left( 1 - \left( 1 - \frac{(p - 1 + 1/q) M_{s,t}}{s m_{2s}} \right)^T \right) f(X^*)\\
		&\ge \frac{1}{p - 1 + 1/q} \frac{m^2_{2s}}{M^2_{s,2}} \left( 1 - \exp\left( -  \frac{(p - 1 + 1/q) M_{s,2} T}{s m_{2s}} \right) \right) f(X^*),
	\end{align*}
	which proves the first statement of the theorem.

	Next, we consider the case where the algorithm stops by finding no pair to improve the objective value.
	For the semi-oblivious and non-oblivious variants, we show that
	\begin{equation*}
		0 \ge q \eta \frac{m_{2s}}{M_{s,t}} f(X^*) - (pq - q + 1) \eta \frac{M_{s,t}}{m_{2s}} f(X)
	\end{equation*}
	holds when the algorithm stops, from which the second statement of the theorem follows.
	When the semi-oblivious variant stops, we have $f(X) \ge f(X - \phi_X(x') + x')$ for all $x' \in N \setminus X$, where $\phi_X(x')$ is defined in the algorithm.
	Hence, in the same manner as the above analysis, we obtain
	\begin{align*}
		0 &\ge \sum_{P \in \calP} \left\{ f((X \cup P) \setminus \phi_X(P \setminus X)) - f(X) \right\} \\
		&\ge q \eta \frac{m_{2s}}{M_{s,t}} f(X^*) - (pq - q + 1) \eta \frac{M_{s,t}}{m_{2s}} f(X).
	\end{align*}
	When the non-oblivious variant stops, we have 
	\begin{equation*}
		0 \ge \frac{1}{2M_{s,t}} \left\| \left( \nabla u(\bfw^{(X)}) \right)_{X' \setminus X} \right\|^2 - \frac{M_{s,t}}{2} \left\| \left( \bfw^{(X)} \right)_{X \setminus X'} \right\|^2 \\
	\end{equation*}
	for all $X' \in \calF_q(X)$.
	Therefore, we have
	\begin{align*}
		0 &\ge \sum_{P \in \calP} \left\{ \frac{1}{2M_{s,t}} \left\| \left( \nabla u(\bfw^{(X)}) \right)_{P \setminus X} \right\|^2 - \frac{M_{s,t}}{2} \left\| \left( \bfw^{(X)} \right)_{P \cap X} \right\|^2 \right\}\\
		&\ge q \eta \frac{m_{2s}}{M_{s,t}} f(X^*) - (pq - q + 1) \eta \frac{M_{s,t}}{m_{2s}} f(X).
	\end{align*}
	by using the above analysis for the first statement.
\end{proof}

\section{Variants with Geometric Improvement}\label{sec:local-geometric}
In this section, we analyze variants of our proposed local search algorithms for sparse optimization under a single matroid constraint.
These variants increase the objective value by at least $1 + \epsilon$ times at each iteration, where $\epsilon > 0$ is a prescribed rate.

To analyze the singleton with the largest objective, which is used as an initial solution for the variants in this section, we use the following fact.
\begin{lemma}\label{lem:initial}
	Suppose $u \colon 2^N \to \bbR$ is a continuously differentiable function with $u(\bfzero) \ge 0$.
	Assume $u$ is restricted strong concave on $\Omega_{s}$ and restricted smooth on $\Omega_{1}$.
	Define $f \colon 2^N \to \bbR$ by $f(X) = \max_{\supp(\bfw) \subseteq X} u(\bfw)$.
	Let $x^* \in \argmax \{ f(x) \mid x \in N \}$.
	We have $f(\{x^*\}) \ge \frac{m_s}{s M_1} f(X)$ for any $X \in \calI$, where $s = \max\{|X| \mid X \in \calI\}$.
\end{lemma}

\begin{proof}
	From the restricted smoothness of $u$, we have
	\begin{equation*}
		u(c_x \bfe_x) 
		\ge u(\bfzero) + \langle \nabla u (\bfzero), c_x \bfe_x \rangle - \frac{M_1}{2} c_x^2
	\end{equation*}
	for any $x \in N$ and $c_x \in \bbR$.
	From \Cref{lem:feature-concave}, we have
	\begin{equation*}
		f(X) - f(\emptyset)
		\le \frac{1}{2m_s} \left\| \left( \nabla u(\bfzero) \right)_X \right\|^2
	\end{equation*}
	for any $X \in \calI$.
	By utilizing these inequalities, for any $X \in \calI$, we have
	\begin{align*}
		\max_{x \in N} f(\{x\})
		&\ge f(\emptyset) + \frac{1}{s} \sum_{x \in X} f(x|\emptyset) \\
		&= f(\emptyset) + \frac{1}{s} \sum_{x \in X} \max_{c_x \in \bbR} \left\{ u(c_x \bfe_x) - u(\bfzero) \right\}\\
		&\ge f(\emptyset) + \frac{1}{s} \sum_{x \in X} \max_{c_x \in \bbR} \left\{ \langle \nabla u (\bfzero), c_x \bfe_x \rangle - \frac{M_1}{2} c_x^2 \right\}\\
		&= f(\emptyset) + \frac{1}{s} \frac{\|  (\nabla u(\bfzero))_X \|^2}{2M_1}\\
		&\ge f(\emptyset) + \frac{m_s}{s M_1} (f(X) - f(\emptyset))\\
		&\ge \frac{m_s}{sM_1} f(X),
	\end{align*}
	which concludes the statement.
\end{proof}

Here we introduce other variants of local search algorithms that use a different type of criteria for finding a pair $(x, x')$ to improve the solution.
These new variants use any pair that increases some function by the rate $(1 + \delta)$, while the previously introduced variants find the pair that yields the largest improvement of some function.
We consider three variants, the oblivious, semi-oblivious, and non-oblivious, similarly to the previous ones.
The oblivious variant searches for any pair $(x, x')$ that increases the objective function by the rate $(1 + \delta)$, that is, $f(X - x + x') \ge (1 + \delta) f(X)$.
The semi-oblivious variant constructs a map $\phi_X \colon N \setminus X \to X$ that satisfies $\phi_X(x') \in \argmin_{x \in X \colon X - x + x' \in \calI} (\bfw^{(X)})_x^2$ and searches for $x' \in N \setminus X$ with $f(x - \phi_X(x') + x') \ge (1 + \delta) f(X)$.
The non-oblivious variant searches for any $(x, x')$ that satisfies
\begin{equation*}
	\frac{1}{2M_{s,2}} \left( \nabla u(\bfw^{(X)}) \right)^2_{x'} - \frac{M_{s,2}}{2} \left(\bfw^{(X)}\right)_{x}^2 \ge \delta f(X).
\end{equation*}
All variants stop when they do not find any solution that satisfies the criteria.
The detailed description of these algorithms is given in \Cref{alg:matroid-delta}.

\begin{algorithm}[t]
	\caption{Local search algorithms for a matroid constraint with geometric improvement}\label{alg:matroid-delta}
	\begin{algorithmic}[1]
		\STATE Let $\delta \gets \epsilon / n$.
		\STATE Let $X \gets \argmax \{ f(v) \mid v \in N \}$.
		\STATE Add arbitrary elements to $X$ until $X$ is maximal in $\calI$.
		\LOOP
			\STATE Search for a pair of $x \in X$ and $x' \in N \setminus X$ such that $X - x + x' \in \calI$ and\\
			$\begin{cases}
				f(X - x + x') \ge (1 + \delta) f(X) & (\text{oblivious})\\
				\displaystyle \text{$f(X - \phi_X(x') + x') \ge (1 + \delta) f(X)$ and $x = \phi_X(x')$,}\\
				\quad \displaystyle \text{where $\phi_X \colon N \setminus X \to X$ is a map that satisfies}\\
				\quad \phi_X(x') \in \argmin_{x \in X \colon X - x + x' \in \calI} (\bfw^{(X)})_x^2 & (\text{semi-oblivious})\\
				\displaystyle \frac{1}{2M_{s,2}} \left( \nabla u(\bfw^{(X)}) \right)^2_{x'} - \frac{M_{s,2}}{2} \left(\bfw^{(X)}\right)_{x}^2 \ge \delta f(X) & (\text{non-oblivious})
			\end{cases}
				$
			\IF{$\exists (x, x')$ satisfying the above condition}
				\STATE Let $X \gets X - x + x'$.
			\ELSE
				\STATE \textbf{return} $X$.
			\ENDIF
		\ENDLOOP
	\end{algorithmic}
\end{algorithm}

We can provide bounds on the approximation ratio of these variants as follows.
\begin{theorem}\label{thm:matroid-geometric}
	Suppose $f(X) = \max_{\supp(\bfw) \subseteq X} u(\bfw)$ and $\calI$ is the independence set family of a matroid.
	Then any of the oblivious, semi-oblivious, and non-oblivious variants of \Cref{alg:matroid-delta} stops after at most $\rmO (\frac{n}{\epsilon} \ln (\frac{s M_1}{m_s}))$ iterations and returns an output $X$ that satisfies
	\begin{equation*}
		f(X) \ge \left( \frac{m_{2s}^2}{M_{s,2}^2} - \epsilon \right) f(X^*),
	\end{equation*}
	where $X^*$ is an optimal solution and $s = \max \{|X| \colon X \in \calI \}$ is the rank of the matroid.
\end{theorem}

\begin{proof}
	Let $X$ be the output of the algorithm.
	Let $X^*$ be an optimal solution.
	From \Cref{lem:schrijver}, we have a bijection $\phi \colon X^* \setminus X \to X \setminus X^*$ such that $X - \phi(x^*) + x^* \in \calI$ for all $x^* \in X^* \setminus X$.
	For each of three variants, we prove
	\begin{equation*}
		0 \ge \sum_{x^* \in X^* \setminus X} \left\{ \frac{1}{2M_{s,2}} \left( \nabla u(\bfw^{(X)}) \right)^2_{x^*} - \frac{M_{s,2}}{2} \left(\bfw^{(X)}\right)_{\phi(x^*)}^2 - \delta f(X) \right\},
	\end{equation*}
	which implies
	\begin{align*}
		0
		&\ge \frac{1}{2M_{s,2}} \left\| \left( \nabla u(\bfw^{(X)}) \right)_{X^*\setminus X} \right\|^2 - \frac{M_{s,2}}{2} \left\| \left(\bfw^{(X)}\right)_{X \setminus X^*} \right\|^2 - \delta n f(X) \\
		&\ge \frac{m_{2s}}{M_{s,2}} f(X^*) - \left(\frac{M_{s,2}}{m_{2s}} + \delta n \right) f(X),
	\end{align*}
	where the second inequality is due to \Cref{lem:feature-concave}.
	Since we set $\delta = \epsilon / n$, we obtain
	\begin{equation*}
		f(X) \ge \left( \frac{m^2_{2s}}{M^2_{s,2}} - \epsilon \right) f(X).
	\end{equation*}

	In the case of the oblivious variant, since $f(X - x + x') \le (1 + \delta) f(X)$ for all $x \in X$ and $x' \in N \setminus X$, we have
	\begin{align*}
		0
		&\ge \sum_{x^* \in X^* \setminus X} \left\{ f(X - \phi(x^*) + x^*) - (1 + \delta) f(X) \right\} \\
		&\ge \sum_{x^* \in X^* \setminus X} \left\{ \frac{1}{2M_{s,2}} \left( \nabla u(\bfw^{(X)}) \right)^2_{x^*} - \frac{M_{s,2}}{2} \left(\bfw^{(X)}\right)_{\phi(x^*)}^2- \delta f(X) \right\}
	\end{align*}
	similarly to the proof of \Cref{thm:matroid-anytime}.
	In the case of the semi-oblivious variant, since $f(X - \phi_X(x') + x')$ for any $x' \in N \setminus X$, we have
	\begin{align*}
		0 &\ge \sum_{x^* \in X^* \setminus X} \left\{ f(X - \phi_X(x^*) + x^*) - (1 + \delta) f(X) \right\} \\
		&\ge \sum_{x^* \in X^* \setminus X} \left\{ \frac{1}{2M_{s,2}} \left( \nabla u(\bfw^{(X)}) \right)^2_{x^*} - \frac{M_{s,2}}{2} \left(\bfw^{(X)}\right)_{\phi_X(x^*)}^2 - \delta f(X) \right\} \\
		&\ge \sum_{x^* \in X^* \setminus X} \left\{ \frac{1}{2M_{s,2}} \left( \nabla u(\bfw^{(X)}) \right)^2_{x^*} - \frac{M_{s,2}}{2} \left(\bfw^{(X)}\right)_{\phi(x^*)}^2 - \delta f(X) \right\}.
	\end{align*}
	When we use the non-oblivious variant, since
	\begin{equation*}
		0 \ge \frac{1}{2M_{s,2}} \left( \nabla u(\bfw^{(X)}) \right)^2_{x^*} - \frac{M_{s,2}}{2} \left(\bfw^{(X)}\right)_{\phi(x^*)}^2 - \delta f(X)
	\end{equation*}
	for all $x^* \in X^* \setminus X$, we obtain
	\begin{equation*}
		0
		\ge \sum_{x^* \in X^* \setminus X} \left\{ \frac{1}{2M_{s,2}} \left( \nabla u(\bfw^{(X)}) \right)^2_{x^*} - \frac{M_{s,2}}{2} \left(\bfw^{(X)}\right)_{\phi(x^*)}^2 - \delta f(X) \right\}.
	\end{equation*}

	Finally, we bound the number of iterations.
	At each iteration, the objective value is improved at least at a rate of $(1 + \delta)$.
	From \Cref{lem:initial}, the initial solution is $\frac{m_s}{s M_1}$-approximation.
	Therefore, the number of iterations is at most $\log_{1+\delta} (\frac{s M_1}{m_s}) = \rmO (\frac{n}{\epsilon} \ln (\frac{s M_1}{m_s}))$.
\end{proof}

To obtain the same bound by using \Cref{thm:matroid-sparse} for \Cref{alg:matroid-anytime}, the number of iterations $T$ is required to be larger than
\begin{equation*}
	T = \frac{s M_{s,2}}{m_{2s}} \log \left( \frac{m_{2s}^2}{\epsilon M_{s,2}^2} \right),
\end{equation*}
which can be larger than \Cref{alg:matroid-delta} in some cases and smaller in other cases.

\begin{remark}
In the same manner, we can devise local search algorithms with geometric improvement for $p$-matroid intersection and $p$-exchange system constraints.
Since they are a straightforward combination of techniques in \Cref{thm:system-anytime} and this section, we omit the description.
\end{remark}

\section{Modular Approximation for Sparse Optimization}\label{sec:modular-approximation}
Modular approximation is a generic method for nonlinear optimization.
This method maximizes a linear function that approximates the original objective function instead of maximizing the original function.
If we can exactly or approximately solve linear function optimization under the sparsity constraint, we can bound the approximation ratio by the restricted strong concavity and restricted smoothness constants.
While we can provide approximation ratio bounds for modular approximation, the empirical performance of modular approximation is mostly poor since it completely ignores correlations between variables in the objective function.

\begin{algorithm}[t]
	\caption{Modular approximation}\label{alg:modular-approximation}
	\begin{algorithmic}[1]
		\STATE Apply the $\alpha$-approximation algorithm to
			\begin{align*}
				\text{Maximize} & \quad \tilde{f}(X) = f(\emptyset) + \sum_{x \in X} f(x | \emptyset)\\
				\text{subject to} & \quad X \in \calI
			\end{align*}
			and obtain the output $X$.
		\STATE \textbf{return} $X$.
	\end{algorithmic}
\end{algorithm}

\begin{proposition}
	Suppose $f(X) = \max_{\supp(\bfw) \subseteq X} u(\bfw)$.
	Assume we use an $\alpha$-approximation algorithm for maximizing a linear function under constraint $\calI$ as a subroutine.
	Modular approximation is $\alpha \frac{m_1 m_s}{M_1 M_s}$-approximation for sparse optimization with constraint $X \in \calI$, where $s = \max \{ |X| \colon X \in \calI \}$.
\end{proposition}

\begin{proof}
	We consider a set function $\tilde{f} \colon 2^N \to \bbR$ defined by
	\begin{equation*}
		\tilde{f}(X) = f(\emptyset) + \sum_{x \in X} f(x | \emptyset)
	\end{equation*}
	for each $X \subseteq N$ to be a modular approximation of $f$.
	From the restricted strong concavity and restricted smoothness of $u$, we have
	\begin{equation}\label{eq:ma-singleton}
		\langle \nabla u (\bfzero), c_i \bfe_i \rangle - \frac{M_1}{2} c_i^2
		\le u(c_i \bfe_i) - u(\bfzero)
		\le \langle \nabla u (\bfzero), c_i \bfe_i \rangle - \frac{m_1}{2} c_i^2
	\end{equation}
	for any $i \in N$ and $c_i \in \bbR$.
	From \Cref{lem:feature-smooth} and \Cref{lem:feature-concave}, we have
	\begin{equation}\label{eq:ma-optimal}
		\frac{1}{2M_s} \left\| \left( \nabla u(\bfzero) \right)_X \right\|^2
		\le f(X) - f(\emptyset)
		\le \frac{1}{2m_s} \left\| \left( \nabla u(\bfzero) \right)_X \right\|^2
	\end{equation}
	for any $X \in \calI$.
	By using \eqref{eq:ma-singleton} and \eqref{eq:ma-optimal}, for any $X \in \calI$, we obtain
	\begin{align*}
	\tilde{f}(X)
	&= u(\bfzero) + \sum_{i \in X} \max_{c_i \in \bbR} \left\{ u(c_i \bfe_i) - u(\bfzero) \right\} \\
	&\le u(\bfzero) + \sum_{i \in X} \max_{c_i \in \bbR} \left\{ \langle \nabla u (\bfzero), c_i \bfe_i \rangle - \frac{m_1}{2} c_i^2 \right\} \tag{from \eqref{eq:ma-singleton}}\\
	&= u(\bfzero) + \sum_{i \in X} \frac{(\nabla u(\bfzero))_i^2}{2m_1}\\
	&\le \frac{M_s}{m_1} f(X) \tag{from \eqref{eq:ma-optimal}},
	\end{align*}
	where we use $u(\bfzero) \ge 0$ for the last inequality.
	Similarly, for any $X \in \calI$, we obtain
	\begin{align*}
		\tilde{f}(X)
		&= u(\bfzero) + \sum_{i \in X} \max_{c_i \in \bbR} \left\{ u(c_i\bfe_i) - u(\bfzero) \right\}\\
		&\ge u(\bfzero) + \sum_{i \in X} \max_{c_i \in \bbR} \left\{ \langle \nabla u (\bfzero), c_i \bfe_i \rangle - \frac{M_1}{2} c_i^2 \right\} \tag{from \eqref{eq:ma-singleton}}\\
		&= u(\bfzero) + \sum_{i \in X} \frac{(\nabla u(\bfzero))_i^2}{2M_1}\\
		&\ge \frac{m_s}{M_1} f(X) \tag{from \eqref{eq:ma-optimal}},
	\end{align*}
	where we use $u(\bfzero) \ge 0$ for the last inequality.
	Let $X_{\textrm{MA}}$ be the output of the $\alpha$-approximation algorithm applied to maximizing $\tilde{f}(X)$ subject to $X \in \calI$.
	Then we have
	\begin{equation*}
		\tilde{f}(X_{\textrm{MA}}) \ge \alpha \tilde{f}(X^*),
	\end{equation*}
	where $X^* \in \argmax_{X \in \calI} f(X)$.
	Finally, we have
	\begin{equation*}
		f(X_{\textrm{MA}}) \ge \frac{m_1}{M_s} \tilde{f}(X_{\textrm{MA}}) \ge \alpha \frac{m_1}{M_s} \tilde{f}(X^*) \ge \alpha \frac{m_1 m_s}{M_1 M_s} f(X^*).
	\end{equation*}
\end{proof}

Since there exists an exact greedy algorithm for maximizing a linear function over a matroid constraint and $(1 / (p - 1 + 1/q) - \epsilon)$-approximation local search algorithms for a $p$-matroid intersection constraint~\citep{LSV10} or $p$-exchange system constraint~\citep{FNSW11}, we obtain the following approximation ratio bounds.

\begin{corollary}
	Modular approximation with the greedy algorithm is $\frac{m_1 m_s}{M_1 M_s}$-approximation for a matroid constraint.
\end{corollary}

\begin{corollary}
	Modular approximation with local search algorithms is $(\frac{1}{p - 1 + 1/q} \frac{m_1 m_s}{M_1 M_s} - \epsilon)$-approximation for a $p$-matroid intersection or $p$-exchange system constraint.
\end{corollary}

\end{document}